\documentclass[a4paper,11pt]{article}
\usepackage{amssymb,amsmath,amsthm,amsfonts,xcolor,hyperref}
\usepackage{subfig}
\usepackage{enumerate}
\usepackage{graphicx,comment}
\usepackage[margin=1in]{geometry}
\newcommand*{\email}[1]{\texttt{#1}} 

\newtheorem{theorem}{Theorem}[section]
\newtheorem{lemma}{Lemma}[section]

\newtheorem{conjecture}{Conjecture}[section]
\newtheorem{proposition}{Proposition}[section]
\newtheorem{claim}{Claim}[section]

\newenvironment{definition}[1][Definition]{\begin{trivlist}
\item[\hskip \labelsep {\bfseries #1}]}{\end{trivlist}}

\def\final{1}  
\def\iflong{\iffalse}
\ifnum\final=0  
\newcommand{\cnote}[1]{[{\tiny Chao: \bf #1}]\marginpar{*}}
\newcommand{\knote}[1]{[{\tiny Karthik: \bf #1}]\marginpar{*}}
\newcommand{\krnote}[1]{[{\tiny Krist\'{o}f: \bf #1}]\marginpar{*}}
\newcommand{\tnote}[1]{[{\tiny Tam\'{a}s: \bf #1}]\marginpar{*}}
\newcommand{\enote}[1]{[{\tiny Euiwoong: \bf #1}]\marginpar{*}}
\newcommand{\todo}[1]{[{\tiny TODO: \bf #1}]\marginpar{*}}
\else 
\newcommand{\cnote}[1]{}
\newcommand{\knote}[1]{}
\newcommand{\krnote}[1]{}
\newcommand{\tnote}[1]{}
\newcommand{\enote}[1]{}
\newcommand{\todo}[1]{}
\fi  

\newcommand{\R}{\mathbb{R}}
\newcommand{\E}{\mathbb{E}}
\newcommand{\N}{\mathbb{N}}

\newcommand{\B}{\mathbb{B}}
\newcommand{\calP}{\mathcal{P}}
\newcommand{\prob}{\mathsf{Pr}}

\newcommand{\Cov}{\mathsf{Cov}}
\newcommand{\Var}{\mathsf{Var}}

\newcommand{\Inf}{\mathsf{Inf}}
\newcommand{\w}{c}
\newcommand{\calD}{\mathcal{D}}
\newcommand{\calG}{\mathcal{D}^{\mathsf{vc}}}
\newcommand{\sfgd}{\mathsf{global}}
\newcommand{\sfst}{\mathsf{st}}

\newcommand{\ug}{\textsc{Unique\-Games}}

\newcommand{\vc}{\textsc{Ver\-tex\-Co\-ver}}
\newcommand{\kvc}{\textsc{Ver\-tex\-Co\-ver on $k$-partite Graphs}}
\newcommand{\threevc}{\textsc{Ver\-tex\-Co\-ver on $3$-partite Graphs}}
\newcommand{\fourvc}{\textsc{Ver\-tex\-Co\-ver on $4$-partite Graphs}}
\newcommand{\kregvc}{\textsc{Ver\-tex\-Co\-ver on $k$-regular Graphs}}

\newcommand{\globaldoublecut}{\textsc{Node\-Double\-Cut}}
\newcommand{\stdoublecut}{\textsc{$\{s,t\}$-Node\-Double\-Cut}}
\newcommand{\nodethreecut}{\textsc{Node-3-Cut}}
\newcommand{\sbicut}{\textsc{$\{s,*\}$-Edge\-Bi\-Cut}}
\newcommand{\edgebicut}{\textsc{Edge\-Bi\-Cut}}
\newcommand{\nodebicut}{\textsc{Node\-Bi\-Cut}}
\newcommand{\srtlinthreecut}{\textsc{$(s,r,t)$-Edge\-Lin\-3\-Cut}}

\newcommand{\stlinthreecut}{\textsc{$(s,*,t)$-Edge\-Lin\-3\-Cut}}
\newcommand{\inoutarbblocker}{\textsc{$r$-Edge\-In\-Out\-Blocker}}
\newcommand{\strongmincut}{\textsc{Strong\-Min\-Cut}}
\newcommand{\stsepkcut}{\textsc{$\{s,t\}$-Sep\-Edge\-$k$\-Cut}}

\newcommand{\VD}{V_\mathcal{D}}
\newcommand{\ED}{E_\mathcal{D}}
\newcommand{\ID}{I_\mathcal{D}}
\newcommand{\bara}{\overline{\alpha}}

\newcommand{\OPT}{\ensuremath{\beta}}
\newcommand{\e}{\epsilon}
\newcommand{\bicut}{\beta}
\newcommand{\cross}{\sigma}
\def\set#1{\{ #1 \}}

\newcommand{\jf}{n_{\mathsf{jf}}}
\newcommand{\jb}{n_{\mathsf{jb}}}
\newcommand{\rf}{n_{\mathsf{rf}}}
\newcommand{\rb}{n_{\mathsf{rb}}}

\title{Global and fixed-terminal cuts in digraphs\footnote{Krist\'{o}f and Tam\'{a}s are supported by the Hungarian National Research, Development and Innovation Office -- NKFIH
grants K109240 and K120254. Chao is supported in part by NSF grant CCF-1526799}}
\date{}
\author{
Krist\'{o}f B\'{e}rczi\thanks{MTA-ELTE Egerv\'{a}ry Research Group, Budapest \email{\{berkri,tkiraly\}@cs.elte.hu}}
\and
Karthekeyan Chandrasekaran\thanks{University of Illinois, Urbana-Champaign  \email{\{karthe,chaoxu3\}@illinois.edu}} 
\and 
Tam\'{a}s Kir\'{a}ly\footnotemark[1]
\and
Euiwoong Lee\thanks{Carnegie Mellon University, Pittsburgh \email{euiwoonl@cs.cmu.edu}}
\and
Chao Xu\footnotemark[2]
}

\begin{document}
\maketitle              
\knote{Suggest titles}
\tnote{Blocking arborescences and paths in both directions: global and fixed-terminal cut problems}
\krnote{Global and fixed-terminal cut problems}
\knote{Global and fixed-terminal cuts in digraphs}
\knote{Double-cut and bicut: global vs fixed-terminal}
\knote{Double-cut and Bicut}
\knote{Arborescence blockers and bicuts}

\begin{abstract}
The computational complexity of multicut-like problems may vary significantly depending on whether the terminals are fixed or not. In this work we present a comprehensive study of this phenomenon in two types of cut problems in directed graphs: double cut and bicut. 
\begin{enumerate}
\item  
Fixed-terminal edge-weighted double cut is known to be solvable efficiently. We show that fixed-terminal node-weighted double cut cannot be 
approximated to a factor smaller than $2$ under the Unique Games Conjecture (UGC), and we also give a 2-approximation algorithm.
For the global version of the problem, we prove an inapproximability bound of $3/2$ under UGC. 
\item 
Fixed-terminal edge-weighted bicut is known to have an approximability factor of $2$ that is tight under UGC. We show that the global edge-weighted bicut is approximable to a factor strictly better than $2$, and that the global node-weighted bicut cannot be approximated to a factor smaller than $3/2$ under UGC.
\item In relation to these investigations, we also prove two results on undirected graphs which are of independent interest. First, we show NP-completeness and a tight inapproximability bound of $4/3$ for the node-weighted $3$-cut problem under UGC.
Second, we show that for constant $k$, there exists an efficient algorithm to solve the
minimum $\{s,t\}$-separating $k$-cut problem. 
\end{enumerate}
Our techniques for the algorithms are combinatorial, based on LPs and based on the enumeration of approximate min-cuts. Our hardness results are based on combinatorial reductions and integrality gap instances.

\end{abstract}

\section{Introduction}
The minimum two-terminal cut problem  (min $s-t$ cut) and its global variant (min cut) are classic interdiction problems with fast algorithms. Generalizations of the fixed-terminal variant, including the multi-cut and the multi-way cut, as well as generalizations of the global variant, including the $k$-cut, have been well-studied in the algorithmic literature \cite{DJPSY94,GH94}. In this work, we study two generalizations of global cut problems to directed graphs, namely double cut and bicut (that we describe below). We study the power and limitations of fixed terminal variants of these cut problems in order to solve the global variants.
In the process, we examine ``intermediate'' multicut problems where only a subset of the terminals are fixed, and obtain results of independent interest. In particular, we show that the undirected $\{s,t\}$-separating $k$-cut problem, where two of the $k$ terminals are fixed, is polynomial-time solvable for constant $k$.
In what follows, we describe the problems along with the results. We refer the reader to Tables \ref{table:global-digraphs}, \ref{table:fixed-digraphs}, and \ref{table:global-undir} at the end of Section \ref{sec:additional} for a summary of the results.
We mention that all our algorithmic/approximation results hold for the min-cost variant while the inapproximability results hold for the min-cardinality variant by standard modification of our reductions and algorithms. For ease of presentation, we do not make this distinction.

The starting point of this work is node-weighted double cut, that we describe below. We recall that an arborescence in a directed graph $D=(V,E)$ is a minimal subset $F\subseteq E$ of arcs such that there exists a node $r\in V$ with every node $u\in V$ having a unique path from $r$ to $u$ in the subgraph $(V,F)$ (e.g., see \cite{schrijver-comb-opt-book}).\\


\noindent{\bf Double Cut.} The input to the \globaldoublecut\ problem is a directed graph and the goal is to find the smallest number of nodes whose deletion ensures that the remaining graph has no arborescence.
\globaldoublecut\ is a generalization of node weighted global min cut in undirected graphs to directed graphs. It is non-monotonic under node deletion.
This problem is key to understanding fault tolerant consensus in networks. We briefly describe this connection.\\

\noindent{\bf Significance of double cut.}
In a recent work, Tseng and Vaidya \cite{TV15} showed that \emph{consensus} in a directed graph can be achieved in the \emph{synchronous model} subject to the failure of $f$ nodes \emph{if and only if} the removal of any $f$ nodes still leaves an arborescence in the remaining graph. Thus, the number of nodes whose failure can be tolerated for the purposes of achieving consensus in a network is \emph{exactly} one less than the smallest number of nodes whose removal ensures that there is no arborescence in the network. So, it is imperative for the network authority to be able to compute this number.\\

A directed graph $D=(V,E)$ has no arborescence if and only if \footnote{We believe that this characterization led earlier authors \cite{BP13} to coin the term \emph{double cut} to refer to the edge deletion variant of the problem. We are following this naming convention.} there exist two distinct nodes $s,t\in V$ such that every node $u\in V$ can reach at most one node in $\{s,t\}$.
By this characterization, every directed graph that does not contain a tournament has a feasible solution to \globaldoublecut .
This characterization motivates the following fixed-terminal variant, denoted \stdoublecut : Given a directed graph with two specified nodes $s$ and $t$, find the smallest number of nodes whose deletion ensures that every remaining node $u$ can reach at most one node in $\{s,t\}$ in the resulting graph. An instance of \stdoublecut\ has a feasible solution provided that the instance has no edge between $s$ and $t$. An efficient algorithm to solve/approximate \stdoublecut\ immediately gives an efficient algorithm to solve/approximate \globaldoublecut .\\


\noindent{\bf Edge-weighted case.} In the edge-weighted version of the problem, \textsc{$\{s,t\}$-EdgeDoubleCut}, the goal is to delete the smallest number of edges to ensure that every node in the graph can reach at most one node in $\{s,t\}$. Similarly, in the global variant, denoted \textsc{EdgeDoubleCut}, the goal is to delete the smallest number of edges to ensure that there exist nodes $s,t$ such that every node $u$ can reach at most one node in $\{s,t\}$, i.e.\
the graph has no arborescence. The fixed-terminal variant \textsc{$\{s,t\}$-EdgeDoubleCut} is solvable in polynomial time using maximum flow and, consequently, \textsc{EdgeDoubleCut} is also solvable in polynomial time (see e.g.\ \cite{BP13}).\\


\noindent{\bf Results for double cut.} Our main result on the fixed-terminal variant, namely \stdoublecut , is the following hardness of approximation.

\begin{theorem}\label{thm:st-node-double-cut-hardness}
\stdoublecut\ is NP-hard,
and has no efficient $(2-\epsilon)$-approximation for any $\epsilon>0$ assuming the Unique Games Conjecture.
\end{theorem}

We also give a $2$-approximation algorithm for \stdoublecut, which leads to a $2$-approximation for the global variant.

\begin{theorem}\label{thm:NodeDoubleCut-2-approx}
There exists an efficient $2$-approximation algorithm for \stdoublecut\ and \globaldoublecut .
\end{theorem}
While we are aware of simple combinatorial algorithms to achieve the $2$-approximation for \stdoublecut, we present an LP-based algorithm
since it also helps to illustrate an integrality gap instance which is the main tool underlying the hardness of approximation
(Theorem \ref{thm:st-node-double-cut-hardness}) for the problem.

Next we focus on the complexity of \globaldoublecut . We note that the NP-hardness of the fixed-terminal variant does not necessarily mean that the global variant is also NP-hard.

\begin{theorem}\label{thm:NodeDoubleCut-hardness}
\globaldoublecut\ is NP-hard,
and has no efficient $(3/2-\epsilon)$-approximation for any $\epsilon>0$ assuming the Unique Games Conjecture.
\end{theorem}

Bicuts offer an alternative generalization of min cut to directed graphs.
The approximability of the fixed-terminal variant of bicut is well-understood while the complexity of the global variant is unknown.
In the following we describe these bicut problems and exhibit a dichotomic behaviour between the fixed-terminal and the global variant.\\

\noindent{\bf Bicut.}
The edge-weighted two-terminal bicut, denoted \textsc{$\{s,t\}$-Edge\-Bi\-Cut}, is the following: Given a directed graph with two specified nodes $s$ and $t$, find the smallest number of edges
whose deletion ensures that $s$ cannot reach $t$ and $t$ cannot reach $s$ in the resulting graph.
Clearly, \textsc{$\{s,t\}$-EdgeBiCut}
is equivalent to $2$-terminal multiway-cut
(the goal in $k$-terminal multiway cut is to delete the smallest number of edges to ensure that the given $k$ terminals cannot reach each other).
This problem has a rich history and has seen renewed interest in the last few months culminating in inapproximability results matching the best-known approximability factor: \textsc{$\{s,t\}$-EdgeBiCut} admits a $2$-factor approximation (by simple combinatorial techniques)
and has no efficient $(2-\epsilon)$-approximation assuming the Unique Games Conjecture \cite{Lee16,CM17}.
In the global variant, denoted \edgebicut , the goal is to find the smallest number of edges whose deletion ensures that there exist two distinct nodes $s$ and $t$ such that $s$ cannot reach $t$ and $t$ cannot reach $s$ in the resulting digraph.

The dichotomy between global cut problems and fixed-terminal cut problems in undirected graphs is well-known.
For concreteness, we recall \textsc{Edge-$3$-Cut} and \textsc{Edge-$3$-way-Cut}. In \textsc{Edge-$3$-Cut},
the goal is to find the smallest number of edges to delete so that the resulting graph has at least $3$ connected components. In \textsc{Edge-$3$-way-Cut}, the input is an undirected graph with $3$ specified nodes and the goal is to find the smallest number of edges to delete so that the resulting graph has at least $3$ connected components with at most one of the $3$ specified nodes in each. While \textsc{Edge-$3$-way-Cut} is NP-hard \cite{DJPSY94}, \textsc{Edge-$3$-Cut} is solvable efficiently \cite{GH94}. However, such a dichotomy is unknown for directed graphs. In particular, it is unknown whether \edgebicut\ is solvable efficiently. Our next result shows evidence of such a dichotomic behaviour. \\

\noindent{\bf Results for bicut.} While \textsc{$\{s,t\}$-EdgeBiCut} is inapproximable to a factor better than $2$ assuming UGC, we show that \edgebicut\ is approximable to a factor strictly better than $2$.

\begin{theorem}\label{thm:bicut-algorithm}
There exists an efficient $(2-1/448)$-approximation algorithm for \edgebicut .
\end{theorem}

We also consider the node-weighted variant of bicut, denoted \nodebicut : Given a directed graph, find the smallest number of nodes whose deletion ensures that there exist nodes $s$ and $t$ such that $s$ cannot reach $t$ and $t$ cannot reach $s$ in the resulting graph. Every directed graph that does not contain a tournament has a feasible solution to \nodebicut . \nodebicut\ is non-monotonic under node deletion,
and it admits a $2$-approximation by a simple reduction to \textsc{$\{s,t\}$-EdgeBiCut}. We show the following inapproximability result.
\begin{theorem}\label{thm:node-bicut-hardness}
\nodebicut\ is NP-hard,
and has no efficient $(3/2-\epsilon)$-approximation for any $\epsilon>0$ assuming the Unique Games Conjecture.
\end{theorem}

We observe that our approximability and inapproximability factors for \globaldoublecut\ and \nodebicut\ coincide---$2$ and $(3/2-\epsilon)$ respectively (Theorems \ref{thm:NodeDoubleCut-2-approx}, \ref{thm:NodeDoubleCut-hardness} and \ref{thm:node-bicut-hardness}). 

\subsection{Additional Results on Sub-problems and Variants} \label{sec:additional}
In what follows, we describe additional results that concern sub-problems in our algorithms/hardness results,
and also variants of these sub-problems which are of independent interest.\\

\noindent{\bf Node weighted 3-Cut.}
We show the NP-hardness of \globaldoublecut\ in Theorem \ref{thm:NodeDoubleCut-hardness} by a reduction from the node-weighted $3$-cut problem in undirected graphs. In the node weighted $3$-cut problem, denoted \nodethreecut , the input is an undirected graph
and the goal is to find the smallest subset of nodes whose deletion leads to at least $3$ connected components in the remaining graph. A classic result of Goldschmidt and Hochbaum \cite{GH94} showed that the edge-weighted variant, denoted \textsc{Edge-3-Cut} (see above for definition)---more commonly known as $3$-cut---is solvable in polynomial time.
Intriguingly, the complexity of \nodethreecut\ remained open until now.
We present the first results on the complexity of \nodethreecut .
\begin{theorem}\label{thm:node-3-cut-hardness}
\nodethreecut\ is NP-hard,
and has no efficient $(4/3-\epsilon)$-approximation for any $\epsilon>0$ assuming the Unique Games Conjecture.
\end{theorem}

The inapproximability factor of $4/3$ mentioned in the above theorem is tight: the $4/3$-approximation factor can be achieved by guessing $3$ terminals that are separated and then using well-known approximation algorithms to solve the resulting node-weighted $3$-terminal cut instance
\cite{GVY04}.\\

\noindent{\bf \stlinthreecut .}
As a sub-problem in the algorithm for Theorem \ref{thm:bicut-algorithm}, we consider the following, denoted \stlinthreecut\ (abbreviating edge-weighted linear $3$-cut): Given a directed graph $D=(V,E)$ and two specified nodes $s,t\in V$, find the smallest number of edges to delete so that there exists a node $r$ with the property that $s$ cannot reach $r$ and $t$, and $r$ cannot reach $t$ in the resulting graph. This problem is a global variant of \srtlinthreecut, introduced in \cite{EJTT14}, where the input specifies three terminals $s,r,t$ and the goal is to find the smallest number of edges whose removal achieves the property above.
A simple reduction from \textsc{Edge-$3$-way-Cut} shows that \srtlinthreecut\ is NP-hard. 
The approximability of \srtlinthreecut\
was studied by Chekuri and Madan \cite{CM17}. They showed that the inapproximability factor coincides with the flow-cut gap of an associated \emph{path-blocking linear program} assuming the Unique Games Conjecture.

There exists a simple combinatorial $2$-approximation algorithm for \srtlinthreecut .
A  $2$-approximation for \stlinthreecut\ can be obtained by guessing the terminal $r$ and using the above-mentioned approximation.
For our purposes,
we need a strictly better than $2$-approximation for \stlinthreecut ; we obtain the following improved approximation factor.
\begin{theorem}\label{thm:s-star-t-lin-3-cut-algorithm}
There exists an efficient $3/2$-approximation algorithm for \stlinthreecut .
\end{theorem}

\noindent{\bf \stsepkcut .}
In contrast to \srtlinthreecut, we do not have a hardness result for \stlinthreecut .
Upon encountering cut problems in directed graphs, it is often insightful to consider the complexity of the analogous problem in undirected graphs.
Our next result shows that the
following analogous problem in undirected graphs is in fact solvable in polynomial time:
given an undirected graph with two specified nodes $s,t$, remove the smallest subset of edges so that the resulting graph has at least $3$ connected components with $s$ and $t$ being in different components.
More generally, we consider \stsepkcut , where the goal is to delete the smallest subset of edges from a given undirected graph so that the resulting graph has at least $k$ connected components with $s$ and $t$ being in different components.
The complexity of \stsepkcut\ was posed as an open problem by Queyranne \cite{Que12}.
We show that \stsepkcut\ is solvable in polynomial-time for every constant $k$.
\begin{theorem}\label{thm:st-k-cut-algorithm}
For every constant $k$, there is an efficient algorithm to solve \stsepkcut .
\end{theorem}

\noindent{\bf \sbicut .}
While Theorem \ref{thm:bicut-algorithm} shows that \edgebicut\ is approximable to a factor strictly smaller than $2$, we do not have a hardness result.
We could prove hardness for the following intermediate problem, denoted \sbicut : Given a directed graph with a specified node $s$, find the smallest number of edges to delete so that there exists a node $t$ such that $s$ cannot reach $t$ and $t$ cannot reach $s$ in the resulting graph. \sbicut\  admits a $2$-approximation by guessing the terminal $t$ and then using the $2$-approximation for \textsc{$\{s,t\}$-EdgeBiCut}. We show the following inapproximability result:
\begin{theorem}\label{thm:single-terminal-bicut-hardness}
\sbicut\ is NP-hard,
and has no efficient $(4/3-\epsilon)$-approximation for any $\epsilon>0$ assuming the Unique Games Conjecture.
\end{theorem}

\begin{table}[h]
    \begin{center}
		\small{
		\begin{tabular}{|c|c|c|}
        \hline
        {\bf Problem} & {\bf Edge-deletion} & {\bf Node-deletion}\\
        \hline
        \textsc{DoubleCut} & {\color{gray} Poly-time} \cite{BP13} & $2$-approx (Thm \ref{thm:NodeDoubleCut-2-approx})\\
        & & $(3/2-\epsilon)$-inapprox (Thm \ref{thm:NodeDoubleCut-hardness})\\
        \hline
        \textsc{BiCut} & $(2-1/448)$-approx (Thm \ref{thm:bicut-algorithm}) & {\color{gray} $2$-approx}\\
        & & $(3/2-\epsilon)$-inapprox (Thm \ref{thm:node-bicut-hardness})\\
        \hline
        \textsc{$(s,*)$-BiCut} & {\color{gray} $2$-approx} & {\color{gray} {\color{gray} $2$-approx}}\\
        & $(4/3-\epsilon)$-inapprox (Thm \ref{thm:single-terminal-bicut-hardness})           & $(3/2-\epsilon)$-inapprox\\ 
        \hline
        \textsc{$(s,*,t)$-Lin$3$Cut} & $3/2$-approx (Thm \ref{thm:s-star-t-lin-3-cut-algorithm}) & {\color{gray} 2-approx} \\
        & & $(4/3 - \epsilon)$-inapprox  \\
        \hline
    \end{tabular}}
    \caption{{\small Global Variants in Directed Graphs. Text in gray refer to known results while text in black refer to the results from this work. All hardness of approximation results are under UGC.
		Hardness results for Node weighted \textsc{$(s,*,t)$-Lin$3$Cut} are based on the fact that it is as hard to approximate as Node weighted \textsc{$\{s,t\}$-Sep$3$Cut} by bidirecting the edges (Table \ref{table:global-undir}).}}
    \vspace{-10pt}
    \label{table:global-digraphs}
    \end{center}
\end{table}


\begin{table}[h]
    \begin{center}
    \small{
		\begin{tabular}{|c|c|c|}
        \hline
        {\bf Problem} & {\bf Edge-deletion} & {\bf Node-deletion}\\
        \hline
        \textsc{$(s,t)$-DoubleCut} & {\color{gray} Poly-time} \cite{BP13} & $2$-approx (Thm \ref{thm:NodeDoubleCut-2-approx})\\
        & & $(2-\epsilon)$-inapprox (Thm \ref{thm:st-node-double-cut-hardness})\\
        \hline
        \textsc{$(s,t)$-BiCut} & {\color{gray} $2$-approx} & {\color{gray} [Equivalent to edge-deletion]}\\
        & {\color{gray} $(2-\epsilon)$-inapprox} \cite{CM16,Lee16} & \\
        \hline
        \textsc{$(s,r,t)$-Lin$3$Cut} & {\color{gray} $2$-approx} & {\color{gray} [Equivalent to edge-deletion]}\\
        & {\color{gray} $(\alpha-\epsilon)$-inapprox } \cite{CM17}& \\
        & {\color{gray} (where $\alpha$ is the flow-cut gap)} & \\
        \hline
    \end{tabular}}
    \caption{\small{Fixed-Terminal Variants in Directed Graphs. Text in gray refer to known results while text in black refer to the results from this work. All hardness of approximation results are under UGC. We include \textsc{$\{s,t\}$-BiCut} and \textsc{$(s,r,t)$-Lin3Cut} for comparison with the global variants in Table \ref{table:global-digraphs}.} }
    \vspace{-10pt}
    \label{table:fixed-digraphs}
    \end{center}
\end{table}

\begin{table}[h]
    \begin{center}
		\small{
    \begin{tabular}{|c|c|c|}
        \hline
        {\bf Problem} & {\bf Edge-deletion} & {\bf Node-deletion}\\
        \hline
        \textsc{$k$-cut} & {\color{gray} Poly-time} \cite{GH94,KS96} & {\color{gray} $(2-2/k)$-approx} \cite{GVY04}\\
        (where $k$ is constant) & & $(2-2/k-\epsilon)$-inapprox (Thm \ref{thm:node-3-cut-hardness})\\
        \hline
        \textsc{$\{s,t\}$-Sep$k$Cut} & Poly-time (Thm \ref{thm:st-k-cut-algorithm}) & {\color{gray} $(2-2/k)$-approx} \cite{GVY04}\\
        (where $k$ is constant) & & $(2-2/k-\epsilon)$-inapprox (Thm \ref{thm:node-3-cut-hardness})\\
        \hline
    \end{tabular}}
    \caption{\small{Global Variants in Undirected Graphs. Text in gray refer to known results while text in black refer to the results from this work. All hardness of approximation results are under UGC.}}
    \vspace{-10pt}
    \label{table:global-undir}
    \end{center}
\end{table}
\newpage


\subsection{Related Work}\label{sec:related-work}
In recent work, Bern\'{a}th and Pap \cite{BP13} studied the problem of deleting the smallest number of arcs to block all minimum cost arborescences of a given directed graph. They gave an efficient algorithm to solve this problem through combinatorial techniques. However, their techniques fail to extend to the node weighted double cut problem.

The node-weighted $3$-cut problem---\textsc{Node-3-Cut}---is a generalization of the classic \textsc{Edge-3-Cut}. Various other generalizations of \textsc{Edge-3-Cut} have been studied in the literature showing the existence of efficient algorithms. These include the edge-weighted $3$-cut in hypergraphs \cite{X10,F10} and the more general submodular $3$-way partitioning \cite{ZNI05,OFN12}.  However, none of these known generalizations address \textsc{Node-3-Cut} as a special case. Feasible solutions to \textsc{Node-3-Cut} are also known as shredders in the node-connectivity literature. In the unit-weight case, shredders whose cardinality is equal to the node connectivity of the graph play a crucial role in the problem of min edge addition to augment node connectivity by one \cite{CT99,J99,LN07,V11}. There are at most linear number of such shredders and all of them can be found efficiently \cite{CT99,J99}. The complexity of finding a min cardinality shredder was open until our results (Theorem \ref{thm:node-3-cut-hardness}). 

In the edge-weighted multiway cut in undirected graphs, the input is an undirected graph with $k$ terminal nodes and the goal is to find the smallest cardinality subset of edges whose deletion ensures that there is no path between any pair of terminal nodes. For $k=3$, a $12/11$-approximation is known \cite{CCT06,KKSTY04}, while for constant $k$, the current-best approximation factor is 1.2975 due to Sharma and Vondr\'{a}k \cite{SV14}. 
These results are based on an LP-relaxation proposed by C\u{a}linescu, Karloff and Rabani \cite{CKR00}, known as the CKR relaxation. Manokaran, Naor, Raghavendra and Shwartz \cite{MNRS08} showed that the inapproximability factor coincides with the integrality gap of the CKR relaxation. Recently, Angelidakis, Makarychev and Manurangsi \cite{AMM16} exhibited instances with integrality gap at least $6/(5+(1/k-1))-\epsilon$ for every $k\ge 3$ and every $\epsilon>0$ for the CKR relaxation.


The node-weighted multiway cut in undirected graphs exhibits very different structure in comparison to the edge-weighted multiway cut. It reduces to edge-weighted multiway cut in hypergraphs. Garg, Vazirani and Yannakakis \cite{GVY04} gave a $(2-2/k)$-approximation for node-weighted multiway cut by exploiting the extreme point structure of a natural LP-relaxation.

The edge-weighted multiway cut in directed graphs has a $2$-approximation, due to Naor and Zosin \cite{NZ01}, as well as Chekuri and Madan \cite{CM16}. Matching inapproximability results were shown recently for $k=2$ \cite{Lee16,CM17}. The node-weighted multiway cut in directed graphs reduces to the edge-weighted multiway cut by exploiting the fact that the terminals are fixed. Such a reduction is unknown for the global version.




\subsection{Preliminaries}
Let $D=(V,E)$ be a directed graph. For two disjoint sets $X,Y\subset V$, we denote $\delta(X,Y)$ to be the set of edges $(u,v)$ with $u\in X$ and $v\in Y$ and $d(X,Y)$ to be the cut value $|\delta(X,Y)|$. We use $\delta^{in}(X):=\delta(V\setminus X, X)$, $\delta^{out}(X):=\delta(X, V\setminus X)$, $d^{in}(X):=|\delta^{in}(X)|$ and $d^{out}(X):=|\delta^{out}(X)|$. We use a similar notation for undirected graphs by dropping the superscripts.
For two nodes $s,t \in V$, a subset $X\subset V$ is called an \emph{$\overline{s}t$-set} if $t\in X\subseteq V-s$. The \emph{cut value} of an $\overline{s}t$-set $X$ is $d^{in}(X)$.

We frequently use the following characterization of directed graphs with no arborescence for the purposes of double cut.
\begin{theorem}(e.g., see \cite{BP13})
Let $D=(V,E)$ be a directed graph. The following are equivalent:
\begin{enumerate}
\item $D$ has no arborescence.
\item There exist two distinct nodes $s,t\in V$ such that every node $u$ can reach at most one node in $\{s,t\}$ in $D$.
\item There exist two disjoint non-empty sets $S,T\subset V$ with $\delta^{in}(S)\cup \delta^{in}(T)=\emptyset$.
\end{enumerate}

\end{theorem}

\section{Approximation for \textsc{NodeDoubleCut}}
\label{sec:nodeDoubleCut-approx}
In this section, we present an efficient $2$-approximation algorithm for \textsc{$\{s,t\}$-NodeDoubleCut} which also leads to a $2$-approximation for \textsc{NodeDoubleCut} by guessing the pair of nodes $s,t$. \\

\noindent{\bf Remark.} Our algorithm is LP-based. Although, alternative combinatorial algorithms can be designed for this problem, we provide an LP-based algorithm since it also helps to illustrate an integrality gap instance which is the main tool underlying the hardness of approximation for the problem. Furthermore, it is also easy to round an \emph{optimum} solution to our LP to obtain a solution whose cost is at most twice the \emph{optimum} LP-cost (using complementary slackness conditions). Here, we present a rounding algorithm which starts from \emph{any feasible solution} to the LP (not necessarily optimal) and gives a solution whose cost is at most twice the LP-cost of \emph{that feasible solution}. 

At the end of this section, we give an example showing that the integrality gap of the LP nearly matches the approximation factor achieved by our rounding algorithm. 

\begin{proof}[Proof of Theorem \ref{thm:NodeDoubleCut-2-approx}]
We recall the problem: Given a directed graph $D=(V,E)$ with two specified nodes $s,t\in V$ and node costs $c:V\setminus \{s,t\}\rightarrow \R_+$, the goal is to find a least cost subset $U\subseteq V\setminus \{s,t\}$ of nodes such that every node $u\in V\setminus U$ can reach at most one node in $\{s,t\}$ in the subgraph $D-U$. We will denote a path $P$ by the set of nodes in the path and the collection of paths from node $u$ to node $v$ by $\calP^{u\rightarrow v}$. 
For a fixed function $d:V\rightarrow \R_+$, the $d$-distance of a path $P$ is defined to be $\sum_{u\in P}d_u$ and the shortest $d$-distance from node $u$ to node $v$ is the minimum $d$-distance among all paths from node $u$ to node $v$. 
We use the following LP-relaxation, where we have a variable $d_u$ for every node $u\in V$:

\begin{align}
\min &\sum_{v\in V\setminus \{s,t\}} c_v d_v \tag{Path-Blocking-LP}\label{lp:path-blocking-lp}\\
\sum_{v\in P} d_v + \sum_{v\in Q} d_v - d_u &\ge 1\ \forall\ P\in \calP^{u\rightarrow s}, Q\in \calP^{u\rightarrow t},\ \forall\ u\in V\notag\\
d_s,d_t&=0\notag\\
d_v&\ge 0\ \forall\ v\in V \notag
\end{align}

We first observe that Path-Blocking-LP can be solved efficiently. The separation problem is the following: given $d:V\rightarrow \R_+$, verify if there exists a node $u\in V$ such that the sum of the shortest $d$-distance path from $u$ to $s$ and the shortest $d$-distance path from $u$ to $t$ is at most $1+d_u$. 
Thus, the separation problem can be solved efficiently by solving the shortest path problem in directed graphs. 

Let $d:V\rightarrow R_+$ be a feasible solution to Path-Blocking-LP. We now present a rounding algorithm that achieves a $2$-factor approximation. We note that our algorithm rounds an arbitrary feasible solution $d$ to obtain an integral solution whose cost is at most twice the LP-cost of the solution $d$. For a subset $U$ of nodes, let $\Delta^{in}(U)$ be the set of nodes $v\in V\setminus U$ that have an edge to a node $u\in U$.

\vspace{1mm}
\hrule
\vspace{1mm}
\noindent \textbf{Rounding Algorithm for $\{s,t\}$-NodeDoubleCut}
\hrule
\begin{enumerate}
\item Pick $\theta$ uniformly from the interval $(0,1/2)$.
\item Let $\B^{in}(s,\theta)$ and $\B^{in}(t,\theta)$ be the set of nodes whose shortest $d$-distance to $s$ and $t$ respectively, is at most $\theta$. 
\item Return $U:=\Delta^{in}(\B^{in}(s,\theta))\cup \Delta^{in}(\B^{in}(t,\theta))$.
\end{enumerate}
\hrule
\vspace{1mm}

The rounding algorithm can be implemented to run in polynomial-time. We first show the feasibility of the solution returned by the rounding algorithm. We use the following claim. 
\begin{claim}
For every $\theta\in (0,1/2)$, we have $\B^{in}(s,\theta)\cap \B^{in}(t,\theta)=\emptyset$.
\end{claim}
\begin{proof}
Say $u\in \B^{in}(s,\theta)\cap \B^{in}(t,\theta)$. Then there exists a path $P\in \calP^{u\rightarrow s}$ and a path $Q\in \calP^{u\rightarrow t}$ such that $\sum_{v\in P} d_v + \sum_{v\in Q} d_v\le 2\theta <1$, a contradiction to the fact that $d$ is feasible for Path-Blocking-LP.
\end{proof}


\begin{claim}
The solution $U$ returned by the algorithm is such that every node $u\in V\setminus U$ can reach at most one node in $\{s,t\}$ in the subgraph $D-U$.
\end{claim}
\begin{proof}
Suppose not. Then there exists $u\in V\setminus U$ that can reach both $s$ and $t$ in $D-U$. If $u\not\in \B^{in}(s,\theta)$, then $u$ cannot reach $s$ in $D-U$ since $\B^{in}(s,\theta)$ has no entering edges in $D-U$. Thus, $u\in \B^{in}(s,\theta)$. Similarly, $u\in \B^{in}(t,\theta)$. However, this contradicts the above claim that $\B^{in}(s,\theta)\cap \B^{in}(t,\theta)=\emptyset$.
\end{proof}

We next bound the expected cost of the solution returned by the rounding algorithm. Let $\bar{d}(v,a)$ denote the shortest $d$-distance from node $v$ to node $a$ in $D$. We use the following claim. 
\begin{claim}
Let $\theta\in (0,1/2)$. If $v\in \Delta^{in}(\B^{in}(s,\theta))$ then $\theta<\bar{d}(v,s)\le \theta+d_v$ and $d_v\neq 0$.
\end{claim}
\begin{proof}
If $\bar{d}(v,s)\le \theta$, then $v\in \B^{in}(s,\theta)$, a contradiction to $v\in \Delta^{in}(\B^{in}(s,\theta))$. If $\bar{d}(v,s)>\theta+d_v$, then $v\not\in \Delta^{in}(\B^{in}(s,\theta))$, a contradiction. If $d_v=0$, then $\theta<\bar{d}(v,s)\le \theta+d_v=\theta$, a contradiction. 
 \end{proof}

\begin{claim}
For every $v\in V$, the probability that $v$ is chosen in $U$ is at most $2d_v$.
\end{claim}
\begin{proof}
The claim holds if $v\in\{s,t\}$. Let us fix $v\in V\setminus \{s,t\}$. By the claim above, if $v\in \Delta^{in}(\B^{in}(s,\theta))$ then $\theta<\bar{d}(v,s)\le \theta+d_v$ and $d_v\neq 0$.
Similarly, if $v\in \Delta^{in}(\B^{in}(t,\theta))$, then $\theta<\bar{d}(v,t)\le \theta+d_v$ and $d_v\neq 0$. Now, the probability that $v$ is in $U$ is at most
\[
\prob\left(\theta\in \left(\bar{d}(v,s)-d_v,\min\{\bar{d}(v,s),1/2\}\right)
\cup
\left(\bar{d}(v,t)-d_v,\min\{\bar{d}(v,t),1/2\}\right)
\right).
\]
Without loss of generality, let $\bar{d}(v,s)\le \bar{d}(v,t)$. We may assume that $d_v>0$ and $\bar{d}(v,s)-d_v<1/2$, since otherwise, the probability that $v$ is in $U$ is $0$ and the claim is proved. Now, by the feasibility of the solution $d$ to Path-Blocking-LP, we have that $\bar{d}(v,s)+\bar{d}(v,t)-d_v\ge 1$ and hence $\bar{d}(v,t)\ge 1/2$. Therefore, 
\begin{align*}
\prob(v\in U)
&\le \prob\left(\theta\in \left(\bar{d}(v,s)-d_v,\min(\bar{d}(v,s),1/2)\right)\right)
+
\prob\left(\theta\in \left(\bar{d}(v,t)-d_v,1/2\right)\right)\\
&=\frac{1}{(1/2)}\left(1/2-\bar{d}(v,s)+d_v + 1/2 - \bar{d}(v,t)+d_v\right)\\
&= 2\left(1-(\bar{d}(v,s)+\bar{d}(v,t)-d_v)+d_v\right)\\
&\le 2d_v.
\end{align*}
The first equality in the above is because $\theta$ is chosen uniformly from the interval $(0,1/2)$ while the last inequality is because of the feasibility of the solution $d$ to Path-Blocking-LP.
 \end{proof}
By the above claim, the expected cost of the returned solution is 
\[
\E\left(\sum_{v\in U} c_v\right) = \sum_{v\in V}\prob(v\in U)c_v \le 2\sum_{v\in V}c_v d_v.
\]
Although our rounding algorithm is a randomized algorithm, it can be derandomized using standard techniques. 

\end{proof}

\begin{figure}
\begin{center}
\includegraphics[scale=0.55]{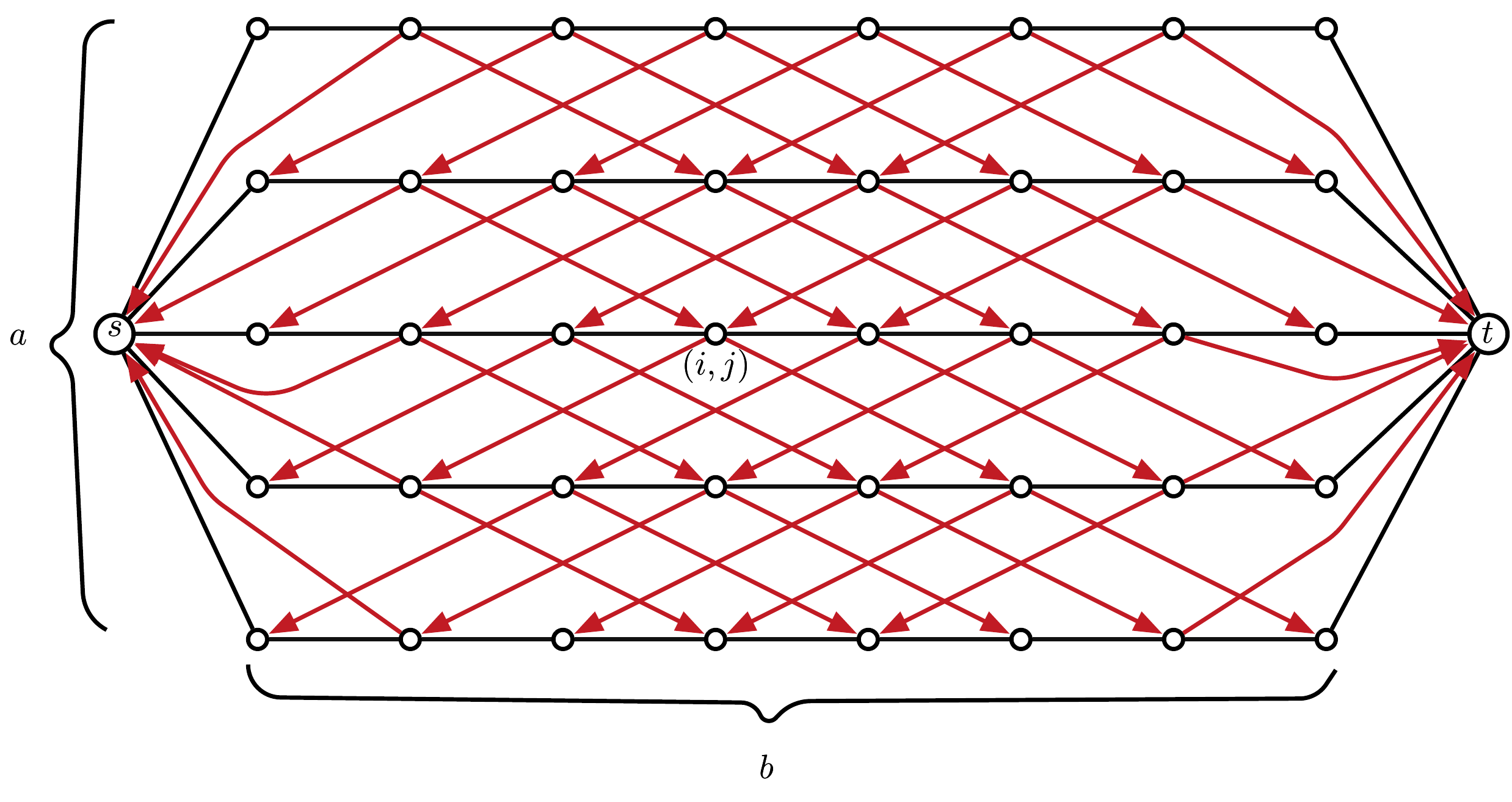}
\end{center}
\caption{$D_{a,b}$ in the proof of Lemma \ref{lem:doublecutintegralitygap} and $(2-\e)$-inapproximability of \stdoublecut.}
\label{fig:pathblockinggap}
\end{figure}

Our next lemma shows a lower bound on the integrality gap that nearly matches the approximation factor achieved by our rounding algorithm. 

\begin{lemma}
\label{lem:doublecutintegralitygap}
The integrality gap of the Path-Blocking-LP for directed graphs containing $n$ nodes is at least $2-7/n^{1/3}$.
\end{lemma}

Our integrality gap instance is also helpful in understanding the hardness of approximation of \stdoublecut. So, we define the instance below and summarize its properties which will be used in the proof of Lemma \ref{lem:doublecutintegralitygap} as well as in the proof of hardness of approximation. 

For two integers $a, b \in \N$, 
consider the directed graph $D_{a, b} = (V_D, A_D)$ obtained as follows (see Figure \ref{fig:pathblockinggap}): Let $V_D := \{ s, t \} \cup ([a] \times [b])$. There are $ab + 2$ nodes. Let $I_D := [a] \times [b]$ and call them as the \emph{internal nodes}. The set of arcs $A_D$ are as follows:
\begin{enumerate}
\item For each $1 \leq i \leq a$, there is a bidirected arc between $s$ and $(i, 1)$, and a bidirected arc between $(i, b)$ and $t$. 
\item For each $1 \leq i \leq a$ and $1 \leq j < b$, there is a bidirected arc between $(i, j)$ and $(i, j + 1)$.
\item For each $1 \leq i < a$ and $2 \leq j \leq b - 1$, there is an arc from $(i, j)$ to $(i + 1, j - 2)$, and an arc from $(i, j)$ to $(i + 1, j + 2)$ (let $(i, 0) := s$ and $(i, b + 1) := t$ for every $i$).
Call them {\em jumping arcs}. 
\end{enumerate}
\begin{lemma}
$D_{a,b}$ has the following properties: 
\begin{enumerate}
\item For each internal node $\alpha = (\alpha_1, \alpha_2) \in I_D$, each $\alpha\rightarrow s$ path has at least $\alpha_2 - a$ internal nodes other than $\alpha$. Similarly, each $\alpha\rightarrow t$ path has at least $b - \alpha_2 - a + 1$ internal nodes other than $\alpha$. 

\item If $S \subseteq I_D$ is such that the subgraph induced by $V_D \setminus S$ has no node $v$ that has paths to both $s$ and $t$, then $|S| \geq 2a - 1$.
\end{enumerate}
\label{lem:skeleton_st_double_cut}
\end{lemma}
\begin{proof}
\begin{enumerate}
\item Jumping arcs are the only arcs that change $\alpha_2$ by $2$ while all other arcs change $\alpha_2$ by $1$. However, a path to $s$ can use at most $a - 1$ jumping arcs because they strictly increase $\alpha_1$. The first property follows from these observations.

\item Suppose that $S \subseteq I_D$ is such that the subgraph induced by $V_D \setminus S$ has no node $v$ that has paths to both $s$ and $t$. 
For $i = 1, \dots, a$, let $s_i := |S \cap \{ \{ i \} \times [b] \}|$. We note that $s_i \geq 1$ for each $i$, otherwise $s$ can reach $t$ and $t$ can reach $s$. 

Suppose $s_i = 1$ for some $1 < i \leq a$ and let $j$ be such that $S \cap \{ \{ i \} \times [b] \} = (i, j)$. If $j = 1$, then $(i, 2)\in V_D\setminus S$ and $(i,2)$ can reach both $s$ and $t$. If $j = b$, then $(i, b - 1)\in V_D\setminus S$ and $(i,b-1)$ can reach both $s$ and $t$. Therefore, we have $1 < j < b$. Then $s_{i - 1} \geq 3$ because $(i - 1, j - 1), (i - 1, j), (i - 1, j + 1)$ can reach both $s$ and $t$ using one jumping arc followed by regular arcs in the $i$th row.

Therefore, $|S| = \sum_{i = 1}^{a} s_i \geq 1 + 2(a - 1) = 2a - 1.$
\end{enumerate}
 \end{proof}

\begin{proof}[Proof of Lemma \ref{lem:doublecutintegralitygap}]
The integer optimum of Path-Blocking-LP on $D_{a,b}$ is at least $2a-1$ by the second property of Lemma~\ref{lem:skeleton_st_double_cut}. Let $r := b - 2a + 1$. We set $d_v := 1/r$ for every internal node $v$. The resulting solution is feasible to Path-Blocking-LP: Indeed, consider $\alpha = (\alpha_1, \alpha_2)$. By the first property of Lemma~\ref{lem:skeleton_st_double_cut}, any $\alpha\rightarrow s$ path and $\alpha\rightarrow t$ path have to together traverse at least $\alpha_2 - a + (b-\alpha_2 - a + 1) = r$ internal nodes. 

Setting $b=a^2$, the integrality gap is at least $(2a-1)/(a^3/r) = 2 - 1/a^3 + 4/a^2 - 5/a \ge 2 - 6/a$ for $a\ge 2$. Using the fact that $a=(|V(D_{a,b})|-2)^{1/3}$, we get the desired bound on the integrality gap.
 \end{proof}

\section{Hardness of Approximation}\label{sec:hardness-of-approx}
In this section, we prove the hardness results, 
namely Theorem~\ref{thm:st-node-double-cut-hardness} for \stdoublecut,
Theorem~\ref{thm:NodeDoubleCut-hardness} for \globaldoublecut,
Theorem~\ref{thm:node-3-cut-hardness} for \nodethreecut, 
Theorem~\ref{thm:single-terminal-bicut-hardness} for \sbicut, and 
Theorem~\ref{thm:node-bicut-hardness} for \nodebicut.
All our reductions begin from\\ $\kregvc$, where the input is an undirected $k$-regular graph, and the goal is to find the smallest subset $S$ of nodes such that every edge in the graph has at least one end-vertex in $S$. 

We use \kvc\ as an intermediate problem, where the input is an undirected $k$-partite graph $G = (V_1 \cup \dots \cup V_k, E)$ (we emphasize that the partitioning $V_1, \dots, V_k$ is  specified explicitly in the input) and the goal is to find the smallest subset $S \subset V_1 \cup \dots \cup V_k$ such that every edge in $E$ has at least one end-vertex in $S$. Our hardness results are structured as follows. 

\begin{enumerate}
\item We first show approximation-preserving (combinatorial) reductions from \kregvc\ (for $k=3$ or $4$) to the above-mentioned problems in Section \ref{sec:reduction_from_vc} (see Fig.~\ref{fig:reductions_overview}). These reductions prove all the inapproximability results under the assumption that $P\neq NP$.
\item For improved hardness of approximation results, we show that \kvc\ is hard to approximate within a factor of $2(k-1)/k-\epsilon$ for any $\epsilon>0$ assuming the Unique Games Conjecture (Section \ref{sec:hardness_vc}). Considering $k=3$ and $k=4$, this result in conjunction with the combinatorial reductions show $(4/3-\epsilon)$-inapproximability for \globaldoublecut\ and \sbicut\, and $(3/2-\epsilon)$-inapproximability for \nodebicut\ assuming the Unique Games Conjecture.
\item We further improve the hardness of approximation for \globaldoublecut\ and \stdoublecut\ by directly reducing from \ug\ via the \emph{length-control dictatorship tests} introduced in \cite{Lee16}. We obtain $(3/2-\epsilon)$-inapproximability for \globaldoublecut\ in Section \ref{sec:hardness_global_double_cut} and\\ $(2 - \epsilon)$-inapproximability for \stdoublecut\ in Section \ref{sec:hardness_st_double_cut}. 
\end{enumerate}

\begin{figure}
    \centering
    {\includegraphics[scale=0.5]{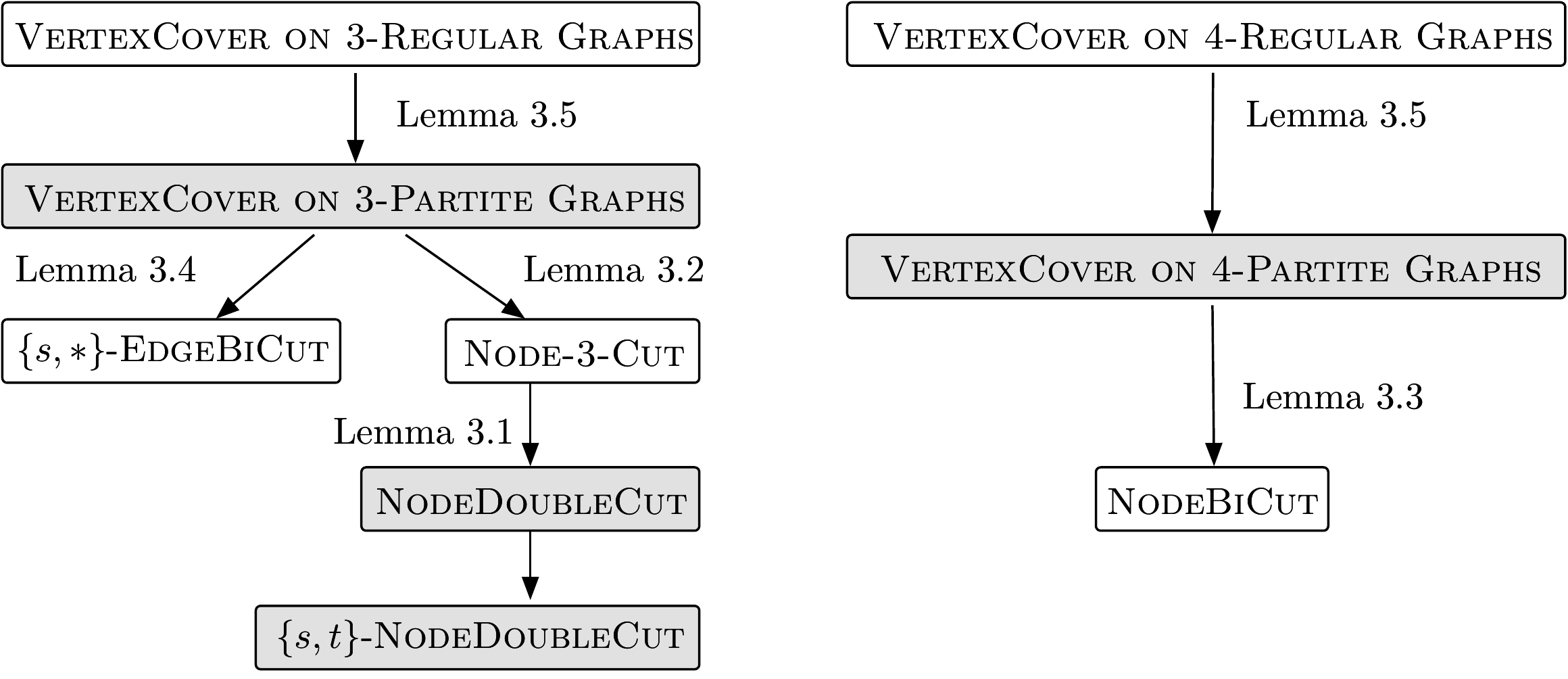} }%
    \qquad
    \caption{Approximation-preserving reductions in Section \ref{sec:reduction_from_vc}. Our NP-hardness results follow from these reductions and the hardness of \vc\ in $3$-regular and $4$-regular graphs~\cite{CC06}. For the problems in shaded cells, we show direct reductions from \ug\ to prove stronger hardness results assuming the Unique Games Conjecture. }%
    \label{fig:reductions_overview}%
\end{figure}
\todo{Remove figure from Appendix.}

We show the combinatorial reductions in Section~\ref{sec:reduction_from_vc}. 
We summarize the preliminaries on discrete Fourier analysis that we need for hardness based on $\ug$ in Section~\ref{sec:ug_prelim}. Similar to most hardness results based on the \ug, the main technical contribution of our work is the construction of {\em dictatorship tests}. In order to avoid repetition, we present the dictatorship tests for the respective problems in Sections~\ref{sec:hardness_global_double_cut},~\ref{sec:hardness_st_double_cut}, and~\ref{sec:hardness_vc}, and show the full reduction from \ug\ to all four problems in Section~\ref{sec:ug}.


\subsection{Combinatorial Reductions}
\label{sec:reduction_from_vc}
\begin{lemma}
For every $\alpha \geq 1$, an $\alpha$-approximation algorithm for \globaldoublecut\ implies an $\alpha$-approximation algorithm for \nodethreecut. 
\label{lem:node_double_cut_to_node_3_cut}
\end{lemma}
\begin{proof}
Given an instance $G = (V, E)$ of \nodethreecut, we use the following algorithm:
\begin{enumerate}
\item For each $s \in V$, 
\begin{enumerate}
\item Construct an instance of \globaldoublecut\ $D = (V, A)$ as follows:
\begin{enumerate}
\item For each edge $\{ u, v \} \in E$, add a bidirected arc between $u$ and $v$.  
\item For each vertex $v \in V \setminus \{ s \}$, add an arc from $v$ to $s$. 
\item Set the weight of $s$ to $\infty$ and the weight of all other vertices to $1$. 
\end{enumerate}
\item Run the $\alpha$-approximation algorithm for \globaldoublecut\ on $D$ and obtain the solution $U_s$. 
\end{enumerate}
\item Output the solution $U_s$ with the smallest size. 
\end{enumerate}
Let $S \subseteq V$ be an optimal solution of \nodethreecut, 
so that $V \setminus S$ is partitioned into $V_1, V_2, V_3$ such that there are no edges between $V_i$ and $V_j$ in $G - S$. 

Consider $s \in V_1$. Let $D$ be the instance of \globaldoublecut\ generated for $s$. We first show that $S$ is a feasible solution for \globaldoublecut\ in $D$. Indeed, consider $u \in V_2$ and $v \in V_3$. In the subgraph of $D$ induced by $V \setminus S$, no vertex can reach both $u$ and $v$, since vertices in $V_i$ have outgoing arcs only to vertices in $V_i$ or to $s$, and $s$ can only reach vertices in $V_1$. 

Therefore, the $\alpha$-approximation algorithm for \globaldoublecut\ will yield $T \subseteq V \setminus \{ s \}$ such that $T \leq \alpha\cdot |S|$
and there exist two vertices $u, v \in V \setminus T$ such that no vertex in $V \setminus T$ can reach both $u$ and $v$. 
Let $U_2$ be the set of vertices strongly connected to $u$, and $U_3$ be the set of vertices strongly connected to $v$. They must be disjoint since otherwise $u$ and $v$ can reach each other. Furthermore, neither $U_2$ nor $U_3$ can contain $s$, since otherwise $v$ can reach $u$ (if $s \in U_2$) or $u$ can reach $v$ (if $s \in U_3$). 

Since every edge of $G$ becomes a bidirected arc in $D$, if some vertices $x$ and $y$ are connected in the subgraph of $G$ induced by $V \setminus T$, then they are strongly connected in the subgraph of $D$ induced by $V \setminus T$. This implies that in the subgraph of $G$ induced by $V \setminus T$, both $U_2$ and $U_3$ are disconnected from the rest of the graph. Therefore, $U_2$ and $U_3$ form a (union of) disjoint connected components, but their union is not $V \setminus T$ (since $s$ is not contained). This implies that $T$ is a feasible solution to $\nodethreecut$ in $G$, finishing the proof. 
 
\end{proof}

\begin{lemma}
There is an approximation-preserving reduction from \threevc\ to \nodethreecut.
\label{lem:node-3-cut-hardness}
\end{lemma}
\begin{proof}
Given a $3$-partite graph $G = (V_1 \cup V_2 \cup V_3, E)$, we construct an instance $G'$ for \nodethreecut\ by adding three vertices $s_1, s_2, s_3$ of infinite weight (all other vertices have weight $1$), and for all $i \in [3]$ and $v \in V_i$, adding an edge between $s_i$ and $v$. Then a subset $S \subseteq V_1 \cup V_2 \cup V_3$ is a vertex cover in $G$ if and only if $G' - S$ has at least three connected components. 
\end{proof}

\begin{lemma}
There is an approximation-preserving reduction from \fourvc\ to \nodebicut.
\label{lem:node-bicut-hardness}
\end{lemma}
\begin{proof}
Given a $4$-partite graph $G = (V_1 \cup V_2 \cup V_3 \cup V_4, E)$, we construct an instance $D = (V_D, A_D)$ for \nodebicut\ as follows: 
Let $V_D := V_1 \cup V_2 \cup V_3 \cup V_4 \cup \{ s, t \}$.
The set of arcs $A_D$ are obtained as follows: 
\begin{enumerate}
\item For every $u, v \in V_i$ for some $i \in [4]$, we add a bidirected arc between $u$ and $v$. 
\item For every $(u, v) \in E$, we add a bidirected arc between $u$ and $v$. 
\item For every $u \in V_1$, we add a bidirected arc between $s$ and $u$.
\item For every $u \in V_2$, we add an arc from $s$ to $u$ and an arc from $t$ to $u$. 
\item For every $u \in V_3$, we add an arc from $u$ to $s$ and an arc from $u$ to $t$. 
\item For every $u \in V_4$, we add a bidirected arc between $t$ and $u$.
\end{enumerate}

We now show the completeness of the reduction. Suppose $R \subseteq V_1 \cup V_2 \cup V_3 \cup V_4$ is a vertex cover in $G$. 
Then $D-R$ has no $s\rightarrow t$ path, since $s$ can only reach vertices in $V_1$ and $V_2$, only vertices in $V_3$ and $V_4$ can reach $t$, and there is no arc between $V_i$ and $V_j$ for any $i \neq j$. Similarly, there is no $t\rightarrow s$ path. Therefore, $R$ is a feasible solution to \nodebicut\ in $D$.

Next we show soundness of the reduction. Suppose $R \subseteq V_1 \cup V_2 \cup V_3 \cup V_4$ is a feasible solution to \nodebicut\ in $D$. There exists two vertices $u, v \in V_D \setminus R$ such that there is no $u\rightarrow v$ path and no $v\rightarrow u$ path in the subgraph of $D$ induced by $V_D \setminus R$. 
We note that $v$ and $u$ cannot be in the same $V_i$ since $V_i$ is a clique in $V_D$. We also rule out the following cases: 
\begin{enumerate}
\item If $v \in V_1, u \in V_2$, then $(v, s, u)$ is a path from $v$ to $u$, a contradiction. 
\item If $v \in V_1, u \in V_3$, then $(u, s, v)$ is a path from $u$ to $v$, a contradiction. 
\item If $v \in V_2, u \in V_3$, then $(u, s, v)$ is a path from $u$ to $v$, a contradiction.  
\item If $v \in V_2, u \in V_4$, then $(u, t, v)$ is a path from $u$ to $v$, a contradiction.  
\item If $v \in V_3, u \in V_4$, then $(v, t, u)$ is a path from $v$ to $u$, a contradiction.  
\end{enumerate}
Thus, $v \in V_1$ and $u \in V_4$. We will show that if $R$ is not a vertex cover, then there is a $v\rightarrow u$ path or $u\rightarrow v$ path, a contradiction. Suppose there exists $\{a, b\} \in E$ such that $a, b \notin R$. 
\begin{enumerate}
\item If $a \in V_1, b \in V_2$, then $(u, t, b, a, v)$ is a path from $u$ to $v$, a contradiction. 
\item If $a \in V_1, b \in V_3$, then $(v, a, b, t, u)$ is a path from $v$ to $u$, a contradiction. 
\item If $a \in V_1, b \in V_4$, then $(v, a, b, u)$ is a path from $v$ to $u$, a contradiction. 
\item If $a \in V_2, b \in V_3$, then $(v, s, a, b, t, u)$ is a path from $v$ to $u$, a contradiction. 
\item If $a \in V_2, b \in V_4$, then $(v, s, a, b, u)$ is a path from $v$ to $u$, a contradiction. 
\item If $a \in V_3, b \in V_4$, then $(u, b, a, s, v)$ is a path from $u$ to $v$, a contradiction. 
\end{enumerate}
Therefore, $R$ must be a vertex cover. This establishes the soundness of the reduction and completes the proof. 
\end{proof}

\begin{lemma}
There is an approximation-preserving reduction from \threevc\ to \sbicut.
\label{lem:single-terminal-bicut-hardness}
\end{lemma}
\begin{proof}
Given a $3$-partite graph $G = (A \cup B \cup C, E)$, we construct an instance $D = (V_D, A_D)$ for \sbicut\  as follows: 
Let $V_D := A_1 \cup A_2 \cup B_1 \cup B_2 \cup C_1 \cup C_2 \cup \{ s, t \}$. For a vertex $v \in A \cup B \cup C$ and $i \in \{ 1, 2 \}$, let $v_i$ denote the corresponding vertex in $V$ (e.g., if $v \in A$, then $v_1 \in A_1$ and $v_2 \in A_2$). 
We introduce three types of arcs in $A_D$.
\begin{enumerate}
\item Vertex arcs: For every $v \in A \cup B \cup C$, create an arc $(v_1, v_2)$ with weight $1$.
\item Forward arcs: Create arcs with weight $\infty$ 
\begin{enumerate}
\item $(s, a_1)$ for all $a \in A$, $(s, b_1)$ for all $b \in B$, $(b_2,s)$ for all $b \in B$, $(c_2, s)$ for all $c \in C$.
\item $(t, a_1)$ for all $a \in A$, $(c_2, t)$ for all $c \in C$.
\item $(a_2, b_1)$ for every $\{a, b\} \in E$, $a\in A, b\in B$ (call them $AB$ arcs), $(a_2, c_1)$ for every $\{a, c\} \in E$, $a\in A, c\in C$ (call them $AC$ arcs), $(b_2, c_1)$ for every $\{b, c\} \in E$, $b\in B, c\in C$ (call them $BC$ arcs). 
\end{enumerate}
\item Backward arcs: Create arcs with weight $\infty$
\begin{enumerate}
\item $(v_2, u_1)$ for all $u, v \in A$ (call them $AA$ arcs), $(v_2, u_1)$ for all $u, v \in C$ (call them $CC$ arcs), $(c_1, a_1)$ for all $a \in A, c \in C$ (call them $CA_1$ arcs), $(c_2, a_2)$ for all $a \in A, c \in C$ (call them $CA_2$ arcs). 
\end{enumerate}
\end{enumerate}

We first show completeness of the reduction. Suppose $R \subseteq A \cup B \cup C$ is a vertex cover in $G$. Let $F = \{ (v_1, v_2) : v \in R \}$. 
We will show that there is no $s\rightarrow t$ path and no $t\rightarrow s$ path in $D-F$.

\begin{enumerate}
\item Suppose there is a $t\rightarrow s$ path in $D-F$. Fix the shortest such $t\rightarrow s$ path $P$. Then, the path $P$ has the following properties: 
\begin{enumerate}
\item Path $P$ does not contain $AA$ arcs or $CA_1$ arcs, since $t$ has direct arcs to vertices in $A_1$. Similarly, $P$ does not contain $CC$ arcs or $CA_2$ arcs, since vertices in $C_2$ have direct arcs to $s$. So, $P$ does not contain any backward arcs.
\item Path $P$ does not contain $BC$ arcs, since vertices in $B_2$ have direct arcs to $s$. 
\end{enumerate}
Thus, the only possibility for the path $P$ is $P = (t, a_1, a_2, v_1, v_2, s)$ for $a \in A$, $v \in B \cup C$, and $\{a, v\} \in E$. This contradicts that $R$ is a vertex cover. 

\item Suppose there is a $s\rightarrow t$ path in $D-F$. Fix the shortest such $t\rightarrow s$ path $P$. Then, the path $P$ has the following properties:
\begin{enumerate}
\item Path $P$ does not contain $AA$ arcs or $CA_1$ arcs, since $s$ has direct arcs to vertices in $A_1$. Similarly, $P$ does not contain $CC$ arcs or $CA_2$ arcs since vertices in $C_2$ have direct arcs $t$. So, $P$ does not contain any backward arcs.
\item Path $P$ does not contain $AB$ arcs, since $s$ has direct arcs to vertices in $B_1$. 
\end{enumerate}
Thus, the only possibility for the path $P$ is $P = (t, v_1, v_2, c_1, c_2, s)$ for $v \in A \cup B$, $c \in C$, and $\{v, c\} \in E$. This contradicts that $R$ is a vertex cover. 
\end{enumerate}
Therefore, $s$ and $t$ cannot reach each other in $D-F$. Consequently, the existence of a vertex cover $R$ in $G$ implies the existence of a feasible solution to \sbicut\ in $D$ of the same size. 

Next we show soundness of the reduction. 
Suppose $F \subseteq E_D$ is a feasible solution to \sbicut\ in $D$. Let $R \subseteq A \cup B \cup C$ be the set of vertices whose vertex arcs are in $F$.
We will show that if $R$ is not a vertex cover in $G$, then every vertex $v \in V_D$ has either a path to $s$ or a path from $s$. 
Since vertices in $A_1, B_1, B_2, C_2$ have a direct arc either from or to $s$, we only need to check vertices in $A_2, C_1$ and $t$. We verify these cases below: 
\begin{enumerate}
\item Suppose there exist $a \in A \setminus R, b \in B \setminus R$ such that $\{a, b\} \in E$. 
\begin{enumerate}[(i)]
\item Considering $t$, we have $(t, a_1, a_2, b_1, b_2, s)$ as a path from $t$ to $s$.
\item For every $a' \in A_2$, we have $(a', a_1, a_2, b_1, b_2, s)$ as a path from $a'$ to $s$.
\item For every $c' \in C_1$, we have $(c', a_1, a_2, b_1, b_2, s)$ as a path from $c'$ to $s$.
\end{enumerate}

\item Suppose there exist $a \in A \setminus R, c \in C \setminus R$ such that $\{a, c\} \in E$. 
\begin{enumerate}[(i)]
\item Considering $t$, we have $(t, a_1, a_2, c_1, c_2, s)$ as a path from $t$ to $s$.
\item For every $a' \in A_2$, we have $(a', a_1, a_2, c_1, c_2, s)$ as a path from $a'$ to $s$. 
\item For every $c' \in C_1$, we have $(c', a_1, a_2, c_1, c_2, s)$ as a path from $c'$ to $s$. 
\end{enumerate}

\item Suppose there exist $b \in B \setminus R, c \in C \setminus R$ such that $\{b, c\} \in E$. 
\begin{enumerate}[(i)]
\item Considering $t$, we have $(s, b_1, b_2, c_1, c_2, t)$ as a path from $s$ to $t$.
\item For every $a' \in A_2$, we have $(s, b_1, b_2, c_1, c_2, a')$ as a path from $s$ to $a'$. 
\item For every $c' \in C_1$, we have $(s, b_1, b_2, c_1, c_2, c')$ as a path from $s$ to $c'$. 
\end{enumerate}
\end{enumerate}
Therefore, the existence of a feasible solution to \sbicut\ in $D$ implies the existence of a vertex cover in $G$ of the same size. This establishes the soundness of the reduction, and proves the lemma.

\end{proof}

The following lemma proves hardness of approximation for \kvc\ for $k = 3, 4$.
We use the result of Chleb{\'\i}k and Chleb{\'\i}kov{\'a}~\cite{CC06} on the hardness of \vc\ on $3$-regular and $4$-regular graphs.

\begin{lemma}
For every $\epsilon > 0$, \threevc\ and \fourvc\ are NP-hard to approximate within a factor $100/99 - \epsilon$ and $53/52 - \epsilon$ respectively. 
\label{lem:hardness_vc}
\end{lemma}
\begin{proof}
Chleb{\'\i}k and Chleb{\'\i}kov{\'a}~\cite{CC06} proved that \vc\ on $3$-regular and $4$-regular graphs are NP-hard to approximate within a factor of $100/99 - \epsilon$ and $53/52 - \epsilon$ respectively for any $\epsilon > 0$. 

For any $k$-regular graph for $k \geq 3$ ($k = 3, 4$ in our case), we recall that Brooks' theorem gives an efficient algorithm to find a $k$-coloring (except when the graph is $K_k$ for which it gives a $(k+1)$-coloring), which gives a $k$-partition of the graph. Therefore, any $\alpha$-approximation algorithm for \kvc\ can be used to get an $\alpha$-approximation for \vc\ on $k$-regular graphs. The lemma follows. 

\end{proof}

\subsection{Preliminaries for Unique Games Hardness}
\label{sec:ug_prelim}

\paragraph*{Gaussian Bounds for Correlated Spaces.}
We introduce the standard tools on correlated spaces from Mossel~\cite{Mossel10}.
Given a probability space $(\Omega, \mu)$ (we always consider finite probability spaces), let $\mathcal{L}(\Omega)$ be the set of functions $\left\{ f : \Omega \rightarrow \R \right\}$ and for an interval $I \subseteq \R$, 
let $\mathcal{L}_I (\Omega)$ be the set of functions $\left\{ f : \Omega \rightarrow I \right\}$.
For two functions $f, g : \Omega \to \R$, we define 
\begin{align*}
\E[f] &:= \E_{x} [f(x)], \\
\Var[f] &:= \E_{x} [(f(x) - \E[f])^2 ],  \\
\Cov[f,g] &:= \E_{x} [(f(x) - \E[f])(g(x) - \E[g])],
\end{align*}
where $x$ is sampled from the probability space $(\Omega, \mu)$. 

 For a subset $S \subseteq \Omega$, we define the {\em measure} of $S$ to be $\mu(S) := \sum_{\omega \in S} \mu(\omega)$. 
A collection of probability spaces are said to be correlated if there is a joint probability distribution on them. We will denote $k$ correlated spaces $\Omega_1,\ldots,\Omega_k$ with a joint distribution $\mu$ as $(\Omega_1 \times \cdots \times \Omega_k , \mu)$. 

Given two correlated spaces $(\Omega_1 \times \Omega_2, \mu)$, we define the correlation between $\Omega_1$ and $\Omega_2$ by 
\[
\rho(\Omega_1, \Omega_2; \mu) := \sup \left\{ \Cov[f, g] : f \in \mathcal{L}(\Omega_1), g \in \mathcal{L}(\Omega_2), \Var[f] = \Var[g] = 1 \right\}.
\]
Given a probability space $(\Omega, \mu)$ and a function $f \in \mathcal{L}(\Omega)$ and $p \in \R^+$, let $\|f\|_{p} := \E_{x \sim \mu}[|f(x)|^p]^{1/p}$. 

Given $(\Omega, \mu)$, let 
$(\Omega^R, \mu^{\otimes R})$ be the {\em product space} where 
for $(x_1, \dots, x_R) \in \Omega^R$, 
$\mu^{\otimes R}(x_1, \dots, x_R) = \prod_{i=1}^R \mu(x_i)$. 
Consider $f \in \mathcal{L}(\Omega^R)$. 
The {\em Efron-Stein decomposition} of $f$ is given by
\[
f(x_1, \dots , x_R) = \sum_{S \subseteq [R]} f_S(x_S)
\]
where (1) $f_S : \Omega^R \to \R$ depends only on $\{ x_i \}_{i \in S}$ and (2) 
for all $S \not \subseteq S'$ and all $x_{S'}$, $
\E_{x' \sim \mu^{\otimes R}} [f_S(x') | x'_{S'} = x_{S'}] = 0$,
where $x_S \in \R^S$ denotes the restriction of $x$ to the coordinates in $S$. 
The {\em influence} of the $i$th coordinate on $f$ is defined by 

\[
\quad \Inf_i[f] := \E_{x_1, \dots , x_{i - 1}, x_{i + 1}, \dots , x_R} [\Var_{x_i} [f(x_1, \dots , x_R)]].
\]
The influence has a convenient expression in terms of the Efron-Stein decomposition. 
\[
\Inf_i[f]  = \| \sum_{S : i \in S} f_S \|_2^2 = 
\sum_{S : i \in S} \| f_S \|_2^2.
\]
We also define the {\em low-degree influence} of the $i$th coordinate.
\[
\mathsf{Inf}_i^{\leq d}[f] := 
\sum_{S : i \in S, |S| \leq d} \| f_S \|_2^2.
\]
For $a, b \in [0, 1]$ and $\rho \in (0, 1)$, let 
\[
\Gamma_{\rho}(a, b) := \Pr[X \leq \Phi^{-1}(a), Y \geq \Phi^{-1}(1 - b)],
\]
where $X$ and $Y$ are $\rho$-correlated standard Gaussian variables and $\Phi$ denotes the cumulative distribution function of a standard Gaussian. 
The following theorem bounds the product of two functions that do not share an influential coordinate in terms of their Gaussian counterparts. 

\begin{theorem}[Theorem 6.3 and Lemma 6.6 of~\cite{Mossel10}]
Let $(\Omega_1 \times \Omega_2, \mu)$ be correlated spaces such that
the minimum nonzero probability of any atom in $\Omega_1 \times \Omega_2$ is at least $\alpha$ and such
that $\rho(\Omega_1, \Omega_2; \mu) \leq \rho$. 
Then for every $\epsilon > 0$ there exist $\tau, d$ depending on $\epsilon$ and $\alpha$ such that if $f : \Omega_1^R \rightarrow [0, 1], g : \Omega_2^R \rightarrow [0, 1]$ satisfy
$
\min(\Inf_i^{\leq d}[f], \Inf_i^{\leq d}[g]) \leq \tau
$
for all $i$, then
$
\E_{(x, y) \in \mu^{\otimes R}} [f(x)g(y)] \geq \Gamma_{\rho}(\E_x[f], \E_y [g]) - \epsilon.
$
\label{thm:mossel}
\end{theorem}

\subsection{$(3/2 - \epsilon)$-Inapproximability for \globaldoublecut}
\label{sec:hardness_global_double_cut}

Consider the directed graph $D = (V_D, A_D)$ (see Figure \ref{fig:nodedoublecut}) defined by 
\begin{align*}
V_D &:= \{ s, t, a, b, c, d \}, \\
A_D &:= \{ (a, s), (s, a), (s, c), (c, a), (a, b), (b, c), (c, b), (d, c), (b, d), (d, t), (t, d), (t, b) \}.
\end{align*} 
Let $I_D := \{ a, b, c, d \}$ be the set of internal vertices. 

\begin{figure}
\begin{center}
\includegraphics[scale=0.5]{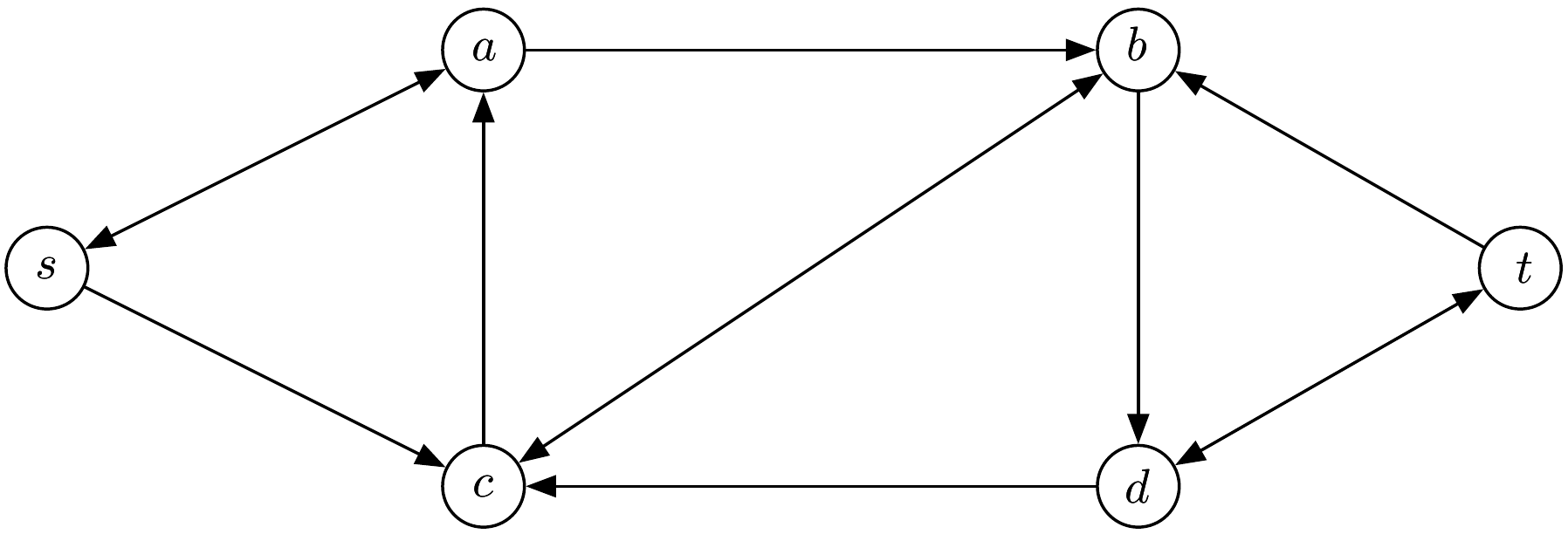}
\end{center} 
\caption{$D$ in the proof of $(3/2-\e)$-inapproximability of \globaldoublecut.}
\label{fig:nodedoublecut}
\end{figure}

We summarize the properties of $D$ that can be verified easily. 
\begin{proposition}
$D$ has the following three properties. 
\begin{enumerate}[(i)]
\item For any vertex $v \in V$, there exists a vertex $u \in \{s, t \}$ such that every $v\rightarrow u$ path has at least three internal vertices. 
\item Every $v \in I_D$ has an incoming arc from either $s$ or $t$. 
\item Even after deleting one vertex from $I_D$, there exists an $s\rightarrow t$ path or a $t\rightarrow s$ path with exactly three remaining internal vertices. 
\end{enumerate}
\label{claim:skeleton_global_double_cut}
\end{proposition}

Based on $D$, we define the {\em dictatorship test} graph $\calD^{\sfgd}_{R, \epsilon} = (V, A)$ as follows, for a positive integer $R$ and $\epsilon > 0$. 
It will be used to show hardness results under the Unique Games Conjecture in Section~\ref{sec:ug}. 
Let $r = 3$. Consider the probability space $(\Omega, \mu)$ where $\Omega := \{ 0, \dots, r-1, * \}$, and $\mu : \Omega \rightarrow [0, 1]$ with $\mu(*) = \epsilon$ and $\mu(x) = (1 - \epsilon)/r$ for $x \neq *$. 

\begin{enumerate}
\item We take $V := \{ s, t \} \cup \{ v^{\alpha}_x \}_{\alpha \in I_D, x \in \Omega^R}$. Let $v^{\alpha}$ denote the set of vertices $\{ v^{\alpha}_x \}_{x \in \Omega^R}$. 

\item For $\alpha \in I_D$ 
and $x \in \Omega^R$, we define the weight as $\w(v^{\alpha}_x) := \mu^{\otimes R}(x)$. We note that the sum of weights is $4$. The terminals $s$ and $t$ have infinite weight. 

\item There are arcs from $s$ to all vertices in $v^{c}$, from $v^a$ to $s$, $s$ to $v^a$, from $v^d$ to $t$, from $t$ to $v^d$, from $t$ to $v^b$. 

\item For each $(\alpha, \beta) \in \{ (c, a), (a, b), (b, c), (c, b), (d, c), (b, d) \}$ and $x, y \in \Omega^R$, we have an arc from $v^{\alpha}_x$ to $v^{\beta}_y$ if
there exists $1 \leq j \leq R$ such that $y_j = (x_j + 1)\mod r$ or $y_j = *$ or $x_j = *$. 
\end{enumerate}

\paragraph*{Completeness.} 
We first prove that removing a set of vertices that correspond to {\em dictators} behaves the same as the fractional solution that gives $1/r$ to every internal vertex. For any $q \in [R]$, let $V_q := \{ v^{\alpha}_x : \alpha \in I_D, x_q = * \mbox{ or } 0 \}$. We note that the total weight of $V_q$ is $4 (\epsilon + (1 - \epsilon)/r) = 4(1 + 2\epsilon)/3$. 

\begin{lemma}
After removing vertices in $V_q$, no vertex in $V$ can reach both $s$ and $t$. 
\label{lem:completeness_global_double_cut}
\end{lemma}
\begin{proof}
Suppose towards contradiction that there exists a vertex that can reach both $s$ and $t$. First, assume that this vertex is $v^{\alpha_0}_{x_0}$ for some $\alpha_0 \in I_D$ and $x_0 \in \Omega^R$. By Property (i) of Proposition~\ref{claim:skeleton_global_double_cut}, there exists $u \in \{ s, t \}$ such that every $\alpha_0\rightarrow u$ path has at least three internal vertices in $D$. Let $(v^{\alpha_0}_{x_0}, v^{\alpha_1}_{x_1}, \dots, v^{\alpha_k}_{x_k}, u)$ be a path from $v^{\alpha}_{x}$ to $u$ in $\calD^{\sfgd}_{R, \epsilon} - V_q$. Note that $k \geq 2$. 

Consider the sequence $((x_0)_q, (x_1)_q, \dots, (x_k)_q)$. 
Recall that $v^{\alpha}_x$ has an arc to $v^{\beta}_y$ for some $\alpha, \beta, x, y$ only if 
$y_q = (x_q + 1) \mod r$ or $y_q = *$ or $x_q = *$. 
Since we removed $V_q$, $(x_i)_q \notin \{ 0, * \}$, $(x_i)_q = (x_{i - 1})_q + 1$. This forces $k \leq 1$, leading to contradiction. 

Finally, assume that $s$ can reach $t$, and
let $(s, v^{\alpha_0}_{x_0}, v^{\alpha_1}_{x_1}, \dots, v^{\alpha_k}_{x_k}, t)$ be a $s\rightarrow t$ path for some $\alpha_i \in I_D$, $x_i \in \Omega^R$. Every $s\rightarrow t$ path in $D$ has to have at least three internal vertices, which forces $k \geq 2$, but considering the sequence $((x_0)_q, (x_1)_q, \dots, (x_k)_q)$ forces $k \leq 1$, which leads to contradiction. Paths from $t$ to $s$ can be ruled out in the same way.
 \end{proof}

\paragraph*{Soundness.}
Suppose that we removed some vertices $C$ such that there exist two vertices $u, v \in V \setminus C$ where no vertex $w \in V \setminus C$ can reach both $u$ and $v$. 
This implies that no vertex $w \in V \setminus C$ can reach both $s$ and $t$, since both $u$ and $v$ have an incoming arc from either $s$ or $t$. Therefore, it suffices to show that unless $C$ reveals an influential coordinate or $\w(C) \geq 2(1 - \epsilon)$, either $s$ can reach $t$ or $t$ can reach $s$. 

To analyze soundness, we define a correlated probability space $(\Omega_1 \times \Omega_2, \nu)$ where both $\Omega_1, \Omega_2$ are copies of $\Omega = \{ 0, \dots, r - 1, * \}$. It is defined by the following process to sample $(x, y) \in \Omega^2$.
\begin{enumerate}
\item Sample $x \in \{0, \dots, r - 1 \}$. Let $y = (x + 1) \mod r$. 
\item Change $x$ to $*$ with probability $\epsilon$. Do the same for $y$ independently. 
\end{enumerate}
We note that the marginal distribution of both $x$ and $y$ is equal to $\mu$. Assuming $\epsilon < 1/2r$, the minimum probability of any atom in $\Omega_1 \times \Omega_2$ is $\epsilon^2$. We use the following lemma to bound the correlation $\rho(\Omega_1, \Omega_2; \nu)$. 

\begin{lemma}[Lemma 2.9 of~\cite{Mossel10}]
Let $(\Omega_1 \times \Omega_2, \mu)$ be two correlated spaces such that the probability of the smallest atom in $\Omega_1 \times \Omega_2$ is at least $\alpha > 0$. Define a bipartite graph $G = (\Omega_1 \cup \Omega_2, E)$ where $a\in \Omega_1, b \in \Omega_2 $ satisfies $\{ a, b \} \in E$ 
if $\mu(a, b) > 0$. If $G$ is connected, then
$
\rho(\Omega_1, \Omega_2 ; \mu) \leq 1 - \alpha^2 / 2.
$
\label{lem:mossel_connected}
\end{lemma}
In our correlated space, the bipartite graph on $\Omega_1 \cup \Omega_2$ is connected since every $x \in \Omega_1$ is connected to $* \in \Omega_2$ and vice versa. Therefore, we can conclude that 
$\rho(\Omega_1, \Omega_2 ; \nu) \leq \rho := 1 - \epsilon^4 / 2$.

Apply Theorem~\ref{thm:mossel} by setting $\rho \leftarrow \rho, \alpha \leftarrow \epsilon^2, \epsilon \leftarrow \Gamma_{\rho}(\epsilon/{3}, \epsilon/{3})/2$ to get $\tau$ and $d$. We will later apply this theorem with the parameters obtained here. 
Fix an arbitrary subset $C \subseteq V$, and let $C_{\alpha}:= C \cap v^{\alpha}$.  
For $\alpha \in I_D$, call $v^{\alpha}$ {\em blocked} if $\mu^{\otimes R} (C_{\alpha}) \geq 1 - \epsilon$. The number of blocked $v^{\alpha}$'s is at most $\w(C) / (1 - \epsilon)$.

By Property (iii) of Proposition~\ref{claim:skeleton_global_double_cut}, unless $\w(C) \geq 2(1 - \epsilon)$ (i.e., unless two vertices are blocked), there exists a path $(s, v^{\alpha_1}, v^{\alpha_2}, v^{\alpha_3}, t)$ or $(t, v^{\alpha_1}, v^{\alpha_2}, v^{\alpha_3}, s)$
where each $v^{\alpha_i}$ is unblocked. Without loss of generality, suppose we have a path 
$(s, v^{\alpha_1}, v^{\alpha_2}, v^{\alpha_3}, t)$. 

For $1 \leq j \leq 3$, let $S_j \subseteq v^{\alpha_j}$ be such that $x \in S_j$ if there exists a path $(v^{\alpha_j}_{x}, 
v^{\alpha_{j + 1}}_{x^{j + 1}}, \dots, v^{\alpha_3}_{x^3}, t)$ for some $x^{j + 1}, \dots, x^3$. 
For $1 \leq j \leq 3$, let $f_j : \Omega^R \rightarrow \{ 0, 1 \}$ be the indicator function of $S_j$. 
We prove that if none of $f_j$ reveals any {\em influential coordinate}, then $\mu^{\otimes R} (S_1) > 0$, which shows that there exists a $s\rightarrow t$ path even after removing vertices in $C$.

\begin{lemma}
Suppose that for any $1 \leq j \leq 3$ and $1 \leq i \leq R$, $\Inf_i^{\leq d}[f_j] \leq \tau$. Then $\mu^{\otimes R}(S_1) > 0$. 
\end{lemma}
\begin{proof}
We prove by induction that $\mu^{\otimes R} (S_j) \geq \epsilon / 3$ for $j = 3, 2, 1$. It holds when $j = 3$ since $v^{\alpha_3}$ is unblocked. Assuming $\mu^{\otimes R} (S_j) \geq \epsilon / 3$, since $S_j$ does not reveal any influential coordinate, Theorem~\ref{thm:mossel} shows that for any subset $T_{j - 1} \subseteq v^{\alpha_{j - 1}}$ with $\mu^{\otimes R}(T_{j - 1}) \geq \epsilon / 3$, there exists an arc from $S_j$ and $T_{j - 1}$. If $S'_{j - 1} \subseteq v^{\alpha_{j - 1}}$ is the set of in-neighbors of $S_j$, we have $\mu^{\otimes R} (S'_{j - 1}) \geq 1 - \epsilon / 3$. Since $v^{\alpha_{j - 1}}$ is unblocked, $\mu^{\otimes R} (S'_{j - 1} \setminus C) \geq 2\epsilon / 3$, completing the induction. 
 \end{proof}

In summary, in the completeness case, if we remove vertices of total weight $\leq {4(1 + 2\epsilon)}/{3}$, no vertex can reach both $s$ and $t$. 
In the soundness case, unless we reveal an influential coordinate or we remove vertices of total weight at least $2(1 - \epsilon)$, there is a $s\rightarrow t$ path or $t\rightarrow s$ path, which means that either $s$ or $t$ can reach every vertex. 
The gap between the two cases is at least \[
\frac{2(1 - \epsilon)}{4(1 + 2\epsilon) / 3 }, \]
which approaches to $3/2$ as $\epsilon \to 0$. 

\subsection{$(2 - \epsilon)$-Inapproximability for \stdoublecut}
\label{sec:hardness_st_double_cut}
Consider the digraph $D_{a,b}$ introduced in Section \ref{sec:nodeDoubleCut-approx}. 
Let $r = b - 2a + 1$. 
Based on $D_{a, b}$, we define the {\em dictatorship test} graph $\calD^{\sfst}_{a, b, R, \epsilon} = (V, A)$ as follows, for a positive integer $R$ and $\epsilon > 0$. 
It will be used to show hardness results under the Unique Games Conjecture in Section~\ref{sec:ug}. 
Consider the probability space $(\Omega, \mu)$ where $\Omega := \{ 0, \dots, r-1, * \}$, and $\mu : \Omega \rightarrow [0, 1]$ with $\mu(*) = \epsilon$ and $\mu(x) = (1 - \epsilon) / r$ for $x \neq *$. 

\begin{enumerate}
\item $V = \{ s, t \} \cup \{ v^{\alpha}_x \}_{\alpha \in I_D, x \in \Omega^R}$. Let $v^{\alpha}$ denote the set of vertices $\{ v^{\alpha}_x \}_{x \in \Omega^R}$. 

\item For $\alpha \in I_D$ and $x \in \Omega^R$, define the weight as $\w(v^{\alpha}_x) = \mu^{\otimes R}(x)$. We note that the sum of weights is $ab$. The terminals $s$ and $t$ have infinite weight. 

\item For each arc between $s$ and $\alpha \in I_D$, for each $x \in \Omega^R$, add an arc with the same direction between $s$ and $v^{\alpha}_x$. Do the same for each arc between $t$ and $\alpha \in I_D$. 

\item For each arc $(\alpha, \beta) \in A_D$ with $\alpha = (\alpha_1, \alpha_2), \beta = (\beta_1, \beta_2) \in I_D$ and $x, y \in \Omega^R$, we have an arc from $v^{\alpha}_x$ to $v^{\beta}_y$ according to the following rule (note that $\alpha_2 \neq \beta_2$).
\begin{enumerate}
\item $\alpha_2 < \beta_2$: add an arc if for any $1 \leq j \leq R$: [$y_j = (x_j + 1)\mod r$] or [$y_j = *$] or [$x_j = *$]. Call them {\em forward} arcs. 
\item $\alpha_2 > \beta_2$: add an arc if for any $1 \leq j \leq R$: [$y_j = (x_j - 1)\mod r$] or [$y_j = *$] or [$x_j = *$]. Call them {\em backward} arcs. 
\item If $(\alpha, \beta) \in A_D$ is a jumping arc, call $(v^{\alpha}_x, v^{\beta}_y)$ also a jumping arc. 
\end{enumerate}
\end{enumerate}

\paragraph*{Completeness.} 
We first prove that removing a set of vertices that correspond to {\em dictators} behaves the same as the fractional solution that gives $1/r$ to every vertex. For any $q \in [R]$, let $V_q := \{ v^{\alpha}_x : \alpha \in I_D, x_q = * \mbox{ or } 0 \}$. We note that the total weight of $V_q$ is \[
ab \left(\epsilon + \frac{1 - \epsilon}{r}\right) \leq ab \epsilon + \frac{ab}{b - 2a}. 
\]

\begin{lemma}
After removing vertices in $V_q$, no vertex in $V$ can reach both $s$ and $t$. 
\label{lem:completeness_st_double_cut}
\end{lemma}
\begin{proof}
Suppose towards contradiction that there exists a vertex that can reach both $s$ and $t$. First, assume that this vertex is $v^{\alpha_0}_{x_0}$ for some $\alpha_0 = ((\alpha_0)_1, (\alpha_0)_2) \in I_D$ and $x_0 \in \Omega^R$. 
Let $p_1 = (v_{x_0}^{\alpha_0}, v_{y_{1}}^{\beta_1}, 
\dots, v_{y_l}^{\beta_{l}}, s)$ be a $v_{x_0}^{\alpha_0}\rightarrow s$ path and $p_2 = (v_{x_0}^{\alpha_0}, v_{x_1}^{\alpha_1}, \dots, v_{x_{k}}^{\alpha_{k}}, t)$ be a $v_{x_0}^{\alpha_0} \rightarrow t$ path in $\calD^{\sfst}_{R, \epsilon} - V_q$ for some $k, l \in \N$,  and $\alpha_1, \dots, \alpha_k, \beta_1, \dots, \beta_l \in I_D$, and $x_1, \dots, x_k, y_1, \dots, y_l \in \Omega^R$. 

\begin{proposition}
$(x_k)_q \geq (x_0)_q + b - (\alpha_0)_2 - a + 1$. 
\end{proposition}
\begin{proof}
Consider the two sequences $(\alpha_0)_q, \dots, (\alpha_k)_q$ and $(x_0)_q, \dots, (x_k)_q$. Since we removed $V_q$, $(\alpha_{i+1})_2 > (\alpha_i)_2$ if and only if $(x_{i + 1})_q > (x_i)_q$. Let $\jf$, $\jb$, $\rf$, $\rb$ be the number forward jumping arcs, backward jumping arcs, forward non-jumping arcs, backward non-jumping arcs in $p_2$ respectively. 
Jumping forward arcs, jumping backward arcs, non-jumping forward arcs, and non-jumping backward arcs change $(\alpha_i)_2$ by $+2, -2, +1$, and $-1$ respectively. 
By considering $(\alpha_0)_q, \dots, (\alpha_k)_q$, 
\[
2\jf + \rf - 2\jb - \rb = (b + 1) - (\alpha_0)_2.
\]
Since using a jumping arc increases $(\alpha_i)_1$ by $1$, 
\[
\jf + \jb \leq a - 1.
\label{eq:st_double_cut_three}
\]
Forward arcs (whether they are jumping or not) increase $(x_i)_q$ by $1$ and backward arc decrease it by $1$. Consider $(x_0)_q, \dots, (x_k)_q$, 
\begin{align*}
(x_k)_q - (x_0)_q & \geq 
\jf + \rf - \jb - \rb - 1 \\
& \geq (2\jf + \rf - 2\jb - 2\rb) - (\jf - \jb)  -  1\\ 
& \geq b - (\alpha_0)_2 - a + 1,
\label{eq:st_double_cut_four}
\end{align*}
as claimed.

\end{proof}
The same proof for $p_1$ shows that $(x_0)_q \geq (y_l)_q + (\alpha_0)_2 - a$. 
Therefore, $(x_k)_q \geq (y_l)_q + b - 2a + 1$ and $(y_l)_q \geq 1$. This implies $(x_k)_q > b - 2a + 1 = r$, leading to contradiction. 
 \end{proof}

\paragraph*{Soundness.}
Suppose that we removed some vertices $C$ such that no vertex $w \in V \setminus C$ can reach both $s$ and $t$. We show this happens only if $C$ reveals an influential coordinate or $\w(C) \geq (2a - 1)(1 - \epsilon)$. 

To analyze soundness, we define a correlated probability space $(\Omega_1 \times \Omega_2, \nu)$ where both $\Omega_1, \Omega_2$ are copies of $\Omega = \{ 0, \dots, r - 1, * \}$. It is defined by the following process to sample $(x, y) \in \Omega^2$.
\begin{enumerate}
\item Sample $x \in \{0, \dots, r - 1 \}$. Let $y = (x + 1) \mod r$. 
\item Change $x$ to $*$ with probability $\epsilon$. Do the same for $y$ independently. 
\end{enumerate}
We note that the marginal distribution of both $x$ and $y$ is equal to $\mu$. Assuming $\epsilon < 1 / 2r$, the minimum probability of any atom in $\Omega_1 \times \Omega_2$ is $\epsilon^2$. By the same arguments as in Section~\ref{sec:hardness_global_double_cut}, 
$\rho(\Omega_1, \Omega_2 ; \nu) \leq \rho := 1 - \epsilon^4 / 2$.

Apply Theorem~\ref{thm:mossel} $\rho \leftarrow \rho, \alpha \leftarrow \epsilon^2, \epsilon \leftarrow \Gamma_{\rho}(\epsilon / 3, \epsilon / 3) / 2 $ to get $\tau$ and $d$. We will later apply this theorem with the parameters obtained here. 
Fix an arbitrary subset $C \subseteq V$, and let $C_{\alpha}:= C \cap v^{\alpha}$.  
For $\alpha \in I_D$, call $v^{\alpha}$ {\em blocked} if $\mu^{\otimes R} (C_{\alpha}) \geq 1 - \epsilon$. The number of blocked $v^{\alpha}$'s is at most $\w(C) / (1 - \epsilon)$.

By Property 2. of Lemma~\ref{lem:skeleton_st_double_cut}, unless $\w(C) \geq (2a - 1)(1 - \epsilon)$ (i.e., unless $2a - 1$ vertices are blocked), there exists $\alpha_0 \in I_D$ and a path $(v^{\alpha_0}, v^{\alpha_{-1}}, \dots, v^{\alpha_{-k}}, s)$ and $(v^{\alpha_0}, v^{\alpha_1},\dots, v^{\alpha_l}, t)$
where each $v^{\alpha_i}$ is unblocked for $-k \leq i \leq l$. 

For $-k \leq j \leq -1$, let $S_j \subseteq v^{\alpha_j}$ be such that $x \in S_j$ if there exists a path $(v^{\alpha_j}_x, v^{\alpha_{j - 1}}_{x^{j - 1}}, \dots, v^{\alpha_{-k}}_{x^{-k}}, s)$ for some $x^{j - 1}, \dots, x^{-k}$. 
Similarly, 
For $1 \leq j \leq l$, let $S_j \subseteq v^{\alpha_j}$ be such that $x \in S_j$ if there exists a path $(v^{\alpha_j}_x, v^{\alpha_{j + 1}}_{x^{j + 1}}, \dots, v^{\alpha_l}_{x^{l}}, t)$ for some $x^{j + 1}, \dots, x^{l}$. 
Let $f_j : \Omega^R \rightarrow \{ 0, 1 \}$ be the indicator function of $S_j$. 
We prove that if none of $f_j$ reveals any {\em influential coordinate}, that there exists a
$x^0 \in \Omega^R$ such that $v^{\alpha_0}_{x^0}$ can reach both $s$ and $t$ even after removing vertices in $C$.


\begin{lemma}
Suppose that for any $j \in \{-k, \dots, -1 \} \cup \{ 1, \dots, l \}$ and $1 \leq i \leq R$, $\Inf_i^{\leq d}[f_j] \leq \tau$. Then there exists a $x^0 \in \Omega^R$ such that $v^{\alpha_0}_{x^0}$ can reach both $s$ and $t$
\end{lemma}
\begin{proof}
We prove that $\mu^{\otimes R} (S_1) \geq \epsilon / 3$ by induction on $j = l, \dots, 1$. 
It holds when $j = l$ since $v^{\alpha_l}$ is unblocked. Assuming $\mu^{\otimes R} (S_j) \geq \epsilon / 3$, since $S_j$ does not reveal any influential coordinate, Theorem~\ref{thm:mossel} shows that for any subset $T_{j - 1} \subseteq v^{\alpha_{j - 1}}$ with $\mu^{\otimes R}(T_{j - 1}) \geq \epsilon / 3$, there exists an arc from $S_j$ and $T_{j - 1}$. If $S'_{j - 1} \subseteq v^{\alpha_{j - 1}}$ is the set of in-neighbors of $S_j$, we have $\mu^{\otimes R} (S'_{j - 1}) \geq 1 - \epsilon / 3$. Since $v^{\alpha_{j - 1}}$ is unblocked, $\mu^{\otimes R} (S'_{j - 1} \setminus C) \geq 2\epsilon / 3$, completing the induction. 

The same argument also proves that $\mu^{\otimes R} (S_{-1}) \geq \epsilon / 3$ by induction on $j = -k, \dots, -1$. 
The total weight of the in-neighbors of $S_{-1}$ in $v^{\alpha_0}$ is at least $1 - \epsilon / 3$, and the total weight of the in-neighbors of $S_{1}$ in $v^{\alpha_0}$ is at least $1 - \epsilon / 3$. Therefore, the total weight of vertices in $v^{\alpha_0}$ that has outgoing arcs to both $S_{-1}$ and $S_1$ is at least $1 - 2\epsilon / 3$. 
Since $\alpha_0$ is not blocked, there exists a vertex $v^{\alpha_0}_{x^0}$ that has outgoing arcs to both $S_1$ and $S_{-1}$, and is not contained $C$. This vertex can reach both $s$ and $t$. 
 \end{proof}

In summary, in the completeness case, if we remove vertices of total weight at most $ab \epsilon + ab / (b - 2a)$, no vertex can reach both $s$ and $t$. 
In the soundness case, unless we reveal an influential coordinate or we remove vertices of total weight at least $(2a - 1)(1 - \epsilon)$, there exists a vertex that can reach both $s$ and $t$. 
The gap between the two cases is at least \[\frac{(2a - 1)(1 - \epsilon)}{ab \epsilon + ab / (b - 2a)}, \] which approaches to $2$ as $a$ increases, $b = a^2$ and $\epsilon = 1 / a^4$.

\subsection{Hardness of \kvc}
\label{sec:hardness_vc}

Fix $k \geq 3$ and $\epsilon > 0$. 
Let $\Omega := \{ *, 0, 1 \}$. 
Let $R \in \N$ be another parameter. 
Our dictatorship test $\calG_{k, R, \epsilon} = ([k] \times \Omega^R, E)$ is defined as follows.
Each vertex is represented by $v^i_{x}$ where $i \in [k]$ and $x \in \Omega^R$ is a $R$-dimensional vector. 
Let $v^i := \{ v^i_x \}_{x \in \Omega^R}$. 
There will be no edge within each $v^i$, so $\calG_{k, R, \epsilon}$ will be $k$-partite. 
Consider the probability space $(\Omega, \mu)$ where $\Omega := \{ 0, 1, * \}$, and $\mu : \Omega \rightarrow [0, 1]$ with $\mu(*) = \epsilon$ and $\mu(x) = (1 - \epsilon) / 2$ for $x \neq *$. 
We define the weight $\w(v^i_x) := \mu^{\otimes R}(x) = \prod_{i=1}^R \mu(x_i)$. The sum of weights is $k$. The edges are constructed as follows. 

\begin{enumerate}
\item There is an edge between $v^i_x$ with $x = (x_1, \dots, x_R)$ and $v^j_y$ with $y = (y_1, \dots, y_R)$ if and only if 
\begin{enumerate}
\item $i \neq j$. 
\item For every $1 \leq l \leq R$: [$x_l \neq y_l$] or [$y_l = *$] or [$x_l = *$]. 
\end{enumerate}
\end{enumerate}

\paragraph*{Completeness.} 
Fix $q \in [R]$ and let $U_{q} := \{ v^i_x : x_q = 0 \mbox{ or } * \}$. The weight  of $U_q$ is $\w(U_q) = k(1 + \epsilon) / 2$. 

\begin{lemma}
$U_q$ is a vertex cover. 
\label{lem:completeness_vc}
\end{lemma}
\begin{proof}
Let $\{ v^i_x, v^j_y \}$ be an edge of $\calG_{k, R, \epsilon}$. If both endpoints do not belong to $U_q$, it implies $x_q = y_q = 1$. It contradicts our construction. 
 \end{proof}

\paragraph*{Soundness.}
To analyze soundness, we define a correlated probability space $(\Omega_1 \times \Omega_2, \nu)$ where both $\Omega_1, \Omega_2$ are copies of $\Omega$. It is defined by the following process to sample $(x, y) \in \Omega^2$.
\begin{enumerate}
\item Sample $x \in \{0, 1 \}$ uniformly at random. Let $y = 1 - x$. 
\item Change $x$ to $*$ with probability $\epsilon$. Do the same for $y$ independently. 
\end{enumerate}
We note that the marginal distribution of both $x$ and $y$ is equal to $\mu$. Assuming $\epsilon < 1/3$, the minimum probability of any atom in $\Omega_1 \times \Omega_2$ is $\epsilon^2$. By the same arguments as in Section~\ref{sec:hardness_global_double_cut}, 
$\rho(\Omega_1, \Omega_2 ; \nu) \leq \rho := 1 - \epsilon^4 / 2$.
Apply Theorem~\ref{thm:mossel} ($\rho \leftarrow \rho, \alpha \leftarrow \epsilon^2, \epsilon \leftarrow \Gamma_{\rho}(\epsilon, \epsilon) / 2 $) to get $\tau$ and $d$. We will later apply this theorem with the parameters obtained here. 

Fix an arbitrary vertex cover $U \subseteq V$, and let $U_{i}:= U \cap v^i$ for $i \in [k]$.  
Let $f_i : \Omega^R \rightarrow \{ 0, 1 \}$ be the indicator function of $U_i$. 
Call $v^i$ {\em blocked} if $\E[f_i] = \mu^{\otimes R} (U_{i}) \geq 1 - \epsilon$. 
The number of blocked $v^i$'s is at most $\w(U) / (1 - \epsilon)$.
We prove that if none of $f_i$ reveals any {\em influential coordinate}, all but one $v^i$'s must be blocked.

\begin{lemma}
Suppose that for any $1 \leq i \leq k$ and $1 \leq j \leq R$, $\Inf_j^{\leq d}[f_i] \leq \tau$.
Then at least $k - 1$ $v^i$'s must be blocked. 
\end{lemma}
\begin{proof}
Assume towards contradiction that there exist $i_1 \neq i_2 \in [k]$ such that $v^{i_1}$ and $v^{i_2}$ are unblocked. 
Since both $f_{i_1}$ and $f_{i_2}$ do not reveal influential coordinates and $\E[f_{i_1}], \E[f_{i_2}] \leq 1 - \epsilon$, 
Theorem~\ref{thm:mossel} ($f \leftarrow 1 - f_1, g \leftarrow 1 - f_2$) shows that 
$\E_{(x, y) \sim \nu^{\otimes R}} [(1 - f_1)(x) \cdot (1 - f_2)(y)]$ is strictly greater than $0$. This implies that there exists $x, y$ such that there is an edge between $v^{i_1}_x$ and $v^{i_2}_y$ but neither $v^{i_1}_x$ nor $v^{i_2}_y$ is contained in $U$. This contradicts that $U$ is a vertex cover. 
 \end{proof}
Therefore, if $U$ does not reveal any influential coordinate, then $\w(U) \geq (k - 1)(1 - \epsilon)$. 
In summary, in the completeness case, there exists a vertex cover of weight $k(1 + \epsilon) / 2$. 
In the soundness case, unless we reveal an influential coordinate, every vertex cover has weight at least $(k - 1)(1 - \epsilon)$. 
The gap between the two cases is at least \[
\frac{2(k - 1)(1 - \epsilon)}{k (1 + \epsilon)},\]
which approaches to $2(k - 1) / k$ as $\epsilon \to 0$.

\subsection{Reduction from \ug}
\label{sec:ug}

\paragraph*{UGC.}

We introduce the Unique Games Conjecture and its equivalent variant.

\begin{definition}
An instance $\mathcal{L}=(B(U_B \cup W_B, E_B),\, [R],\, \left\{ \pi(u, w) \right\}_{(u, w) \in E_B})$ of \ug\ consists of a biregular bipartite graph $B(U_B \cup W_B, E_B)$ and a set $[R]$ of labels. For each edge $(u, w) \in E_B$ there is a constraint specified by a permutation $\pi(u, w) : [R] \rightarrow [R]$. The goal is to find a labeling $l : U_B \cup W_B \rightarrow [R]$ of the vertices such that as many edges as possible are satisfied, where an edge $e = (u, w)$ is said to be satisfied if $l(u) = \pi(u, w)(l(w))$. 
\end{definition}
\begin{definition}
Given a \ug\ instance $\mathcal{L} =(B(U_B \cup W_B, E_B),\, \allowbreak  [R],\, \allowbreak \left\{ \pi(u, w) \right\}_{(u, w) \in E_B})$, let $\mathsf{OPT}(\mathcal{L})$ denote the maximum fraction of simultaneously-satisfied edges of $\mathcal{L}$ by any labeling, i.e.
\begin{equation*}
\mathsf{OPT}(\mathcal{L}) := \frac{1}{|E_B|} \ \max_{l : U_B \cup W_B \rightarrow [R] } |\left\{ e \in E : l \mbox{ satisfies }e \right\} |.
\end{equation*}
\end{definition}
\begin{conjecture} [The Unique Games Conjecture~\cite{Khot02}] 
For any constants $\eta > 0$, there is $R = R(\eta)$ such that, for a \ug\ instance $\mathcal{L}$ with label set $[R]$, it is NP-hard to distinguish between
\begin{enumerate}
\item $\mathsf{OPT}(\mathcal{L}) \geq 1 - \eta$. 
\item $\mathsf{OPT}(\mathcal{L}) \leq \eta$. 
\end{enumerate}
\label{conj:ug}
\end{conjecture}
To show the optimal hardness result for \textsc{VertexCover}, Khot and Regev~\cite{KR08} introduced the following seemingly stronger conjecture, and proved that it is in fact equivalent to the original Unique Games Conjecture.

\begin{conjecture} [Khot and Regev~\cite{KR08}] 
\label{conj:ug_variant}
For any constants $\eta > 0$, there is $R = R(\eta)$ such that, for a \ug\ instance $\mathcal{L}$ with label set $[R]$, it is NP-hard to distinguish between
\begin{enumerate}
\item There is a set $W' \subseteq W_B$ such that $|W'| \geq (1 - \eta)|W_B|$ and a labeling $l : U_B \cup W_B \rightarrow [R]$ that satisfies every edge $(u, w)$ for $v \in U_B$ and $w \in W'$.
\item $\mathsf{OPT}(\mathcal{L}) \leq \eta$. 
\end{enumerate}
\label{conj:ug1}
\end{conjecture}

\paragraph*{General Reduction.}
We now introduce our reduction from \ug\ to our problems \globaldoublecut, \stdoublecut, and \kvc. 
We constructed three dictatorship tests for 
$\calD^{\sfgd}_{R, \epsilon}$, $\calD^{\sfst}_{a, b, R, \epsilon}$, $\calG_{k, R, \epsilon}$.
The first two are directed and $\calG_{k, R, \epsilon}$ is undirected, but they are all vertex-weighted. 
Fix a problem and the parameters, and let $\calD = (\VD, \ED) $ be the dictatorship test with the weight function $c : \VD \to \R$.

Given an instance $\mathcal{L}$ of \ug, we describe how to reduce it to a graph $G = (V_G, E_G)$. 
$G$ will be directed or undirected depending on the problem we reduce to. 
We assign to each vertex $w \in W_B$ a copy of $\VD$ and for each {\em terminal} of $\VD$, merge all $|W_B|$ copies into one. 
The merged terminals are $\{ s, t \}$ for \globaldoublecut\ and \stdoublecut, and \vc\ has no terminal. 
For any $w \in W_B, v \in \VD$, the vertex weight of $(w, v)$  is $\w(v) / |W_B|$, so that the sum of vertex weights (except terminals) is $4$ for \globaldoublecut, $ab$ for \stdoublecut, and $k$ for \kvc. Let $\ID := \VD \setminus \{s, t \}$ for \globaldoublecut\ and \stdoublecut, and $\ID := \VD$ for \kvc. 

For a permutation $\sigma : [R] \rightarrow [R]$, let $x \circ \sigma := (x_{\sigma(1)}, \dots, x_{\sigma(R)})$. 
To describe the set of edges, consider the random process where $u \in U_B$ is sampled uniformly at random, and its two neighbors $w^1, w^2$ are independently sampled. 
For each edge $(v^{i_1}_{x^1}, v^{i_2}_{x^2}) \in E_{\calD}$, we create an edge $((w_1, v^{i_1}_{x^1 \circ \pi(u, w^1)}), (w_2, v^{i_2}_{x^2 \circ \pi(u, w^2)}))$. Call this edge is {\em created by} $u$. 
For each edge incident on a terminal (i.e., $(X, v^i_x)$ or $(v^i_x, X)$ where $X \in \{ s, t \}$), we add the corresponding edge $(X, (w, v^i_x))$ or $((w, v^i_x), X)$ for each $w \in W_B$.

\paragraph*{Completeness.}
Suppose there exists a labeling $l$ and a subset $W' \subseteq W_B$ with $|W'| \geq (1 - \eta)|W_B|$ such that $l$ satisfy every edge incident on $W'$. 

\paragraph*{\globaldoublecut.}
Let $D = (V_D, A_D)$ be the graph constructed in Section~\ref{sec:hardness_global_double_cut} and $I_D$ be $V_D \setminus \{ s, t \}$. 
For every $w \in W'$, we remove the following vertices.
\[
\{ (w, v^{\alpha}_x) : \alpha \in I_D, x_{l(w)} = * \mbox { or } 0 \}. 
\]
For $w \notin W'$, we remove every vertex in $\{ w \} \times \ID$. 
The total weight is at most $4(1 + 2 \epsilon) / 3 + 4\eta$. 
The completeness analysis for the dictatorship test ensures that no vertex in $V_G$ can reach both $s$ and $t$. 
The proof of Lemma~\ref{lem:completeness_global_double_cut} works verbatim --- for 
each vertex $(w_j, v^{\alpha_j}_{x_j})$ with $x_j \in \Omega^R$, consider $(x_j)_{l(w_j)}$ in place of $(x_j)_q$. 

\paragraph*{\stdoublecut.}
Let $D = (V_D, A_D)$ be the graph constructed in Section~\ref{sec:hardness_st_double_cut} and $I_D$ be $V_D \setminus \{ s, t \}$. 
For every $w \in W'$, we remove the following vertices.
\[
\{ (w, v^{\alpha}_x) : \alpha \in I_D, x_{l(w)} = * \mbox { or } 0 \}. 
\]
For $w \notin W'$, we remove every vertex in $\{ w \} \times \ID$. 
The total weight is at most 
$ab / (b - 2a) + ab \epsilon + ab \eta$. 
The completeness analysis for the dictatorship test ensures that no vertex in $V_G$ can reach both $s$ and $t$. 
The proof of Lemma~\ref{lem:completeness_st_double_cut} works verbatim --- for 
each vertex $(w_j, v^{\alpha_j}_{x_j})$ with $x_j \in \Omega^R$, consider $(x_j)_{l(w_j)}$ in place of $(x_j)_q$. 

\paragraph*{\kvc.}
For every $w \in W'$, we remove the following vertices.
\[
\{ (w, v^{\alpha}_x) : \alpha \in [k], x_{l(w)} = * \mbox { or } 0 \}. 
\]
For $w \notin W'$, we remove every vertex in $\{ w \} \times \VD$. 
The total weight is at most $k(1 + \epsilon) / 2 + k\eta$. 
The completeness analysis for the dictatorship test, Lemma~\ref{lem:completeness_vc}, ensures that every edge of $G$ is covered --- for each edge $\{ (w, v^{i}_{x}), (w', v^{j}_{y}) \}$, consider $x_{l(w)}$ and $y_{l(w')}$ in place of $x_q$ and $y_q$.

\paragraph*{Soundness.}
We present the soundness analysis. We first discuss how to extract an influential coordinate for each $u \in U_B$.

\paragraph*{\globaldoublecut\ and \stdoublecut.}
Fix an arbitrary $C \subseteq V_G \setminus \{ s, t \}$, and consider the graph after removing vertices in $C$. 
We will show that if $\w(C)$ is small and no vertex can reach both $s$ and $t$, we can decode influential coordinates for many vertices of the \ug\ instance. 
For \globaldoublecut, since every vertex in $V_G$ has an incoming arc from either $s$ or $t$, it implies that any solution to \globaldoublecut\ must reveal influential coordinates or $\w(C)$ must be large. 
Recall the graph $D = (V_D, A_D)$ constructed in Section~\ref{sec:hardness_global_double_cut} (for \globaldoublecut) or 
Section~\ref{sec:hardness_st_double_cut} (for \stdoublecut), and $I_D = V_D \setminus \{ s, t \}$. 

For each $w \in W_B$, $r \in \{ s, t \}$, $1 \leq j \leq |I_D|$, and a sequence $\bara = (\alpha_1, \dots, \alpha_j) \in (I_D)^j$, let $g_{w, r, j, \bara} : \Omega^R \rightarrow \{ 0, 1 \}$ such that $g_{w, i, j, \bara}(x) = 1$ if and only if there exists a path $p = ((w, v^{\alpha_1}_x), (w_2, v^{\alpha_2}_{x^2}), \dots, (w_{j}, v^{\alpha_{j}}_{x^{j}}), r)$ for some $w_2, \dots, w_{j} \in W_B$ and $x^2, \dots, x^j \in \Omega^R$.

For $u \in U_B, 1 \leq j \leq |I_D|$, and $\bara \in (I_D)^j$, let  $f_{u, r, j, \bara} : \Omega_R \rightarrow [0, 1]$ be such that
\[
f_{u, r, j, \bara} (x) = \E_{w \in N(u)} [g_{w, r, j, \bara} (x \circ \pi^{-1} (u, w)) ],
\]
where $N(u)$ is the set of neighbors of $u$ in the \ug\  instance. 

Let $S := 2(1 - \epsilon)$ (for \globaldoublecut) or $S := (2a - 1)(1 - \epsilon)$ (for \stdoublecut) be the lower bound on the weight in the soundness analysis of the respective dictatorship tests. Let $S' := (1 - \beta)S$ for some $\beta > 0$ that will be determined later,
and assume that the total weight of removed vertices is at most $S'$. 
Let $\gamma(u)$ be the expected weight of $C \cap (\{ w \} \times \ID)$, where $w$ is a random neighbor of $u$. 
Since the instance of \ug\ is biregular, 
\[
\E_{u \in U_B} [\gamma(u)] = \E_{u \in U_B} [\E_{w \in N(u)} [ \w(C \cap (\{ w \} \times \ID))]] = \E_{w \in W_B} [\w(C \cap (\{ w \} \times \ID))]
\leq S'  = (1 - \beta)S.
\]
Therefore, at least $\beta$ fraction of $u$'s have $\gamma(u) \leq S$. 
For such $u$, since no vertex can reach both $s$ and $t$, the soundness analysis for the dictatorship test shows that there exists $q \in [R], r \in \{s, t\}, 1 \leq j \leq |I_D|, \bara$ such that $\Inf_{q}^{\leq d} [f_{u, r, j, \bara}] \geq \tau$ ($d$ and $\tau$ do not depend on $u$).

\paragraph*{\kvc.}
Fix an arbitrary $C \subseteq V_G$, and consider the graph after removing vertices in $C$. 
We will show that if $\w(C)$ is small and every edge is removed, we can decode influential coordinates for many vertices of the \ug\ instance. 

For each $w \in W_B$ and $j \in [k]$, let $g_{w, j} : \Omega^R \rightarrow \{ 0, 1 \}$ such that $g_{w, j}(x) = 1$ if and only if $(w, v^j_x) \notin C$. 
For $u \in U_B$ and $1 \leq j \leq [k]$, let  $f_{u, j} : \Omega_R \rightarrow [0, 1]$ be such that
\[
f_{u, j} (x) = \E_{w \in N(u)} [g_{w, j} (x \circ \pi^{-1} (u, w)) ],
\]
where $N(u)$ is the set of neighbors of $u$ in the \ug\ instance. 

Let $S := (1 - \epsilon)(k - 1)$ be the lower bound on the weight in the soundness analysis of the dictatorship test. Let $S' := (1 - \beta)S$ for some $\beta > 0$ that will be determined later, and assume that the total weight of removed vertices is at most $S'$.
Let $\gamma(u)$ be the expected weight of $C \cap (\{ w \} \times \VD)$, where $w$ is a random neighbor of $u$. 
Since the instance of \ug\ is biregular, 
\[
\E_{u \in U_B} [\gamma(u)] = \E_{u \in U_B} [\E_{w \in N(u)} [ \w(C \cap (\{ w \} \times \ID))]] = \E_{w \in W_B} [\w(C \cap (\{ w \} \times \ID))]
\leq S'  = (1 - \beta)S.
\]
Therefore, at least $\beta$ fraction of $u$'s have $\gamma(u) \leq S$. 
For such $u$, since every edge is removed, the soundness analysis for the dictatorship test shows that there exists $q \in [R], 1 \leq j \leq [k]$ such that $\Inf_{q}^{\leq d} [f_{u, j}] \geq \tau$ ($d$ and $\tau$ do not depend on $u$).

\paragraph*{Finishing Up.} 
The above analyses of \globaldoublecut, \stdoublecut, \kvc\ can be abstracted as follows. Each vertex $u \in U_B$ is associated with functions $\{ f_{u, h} : \Omega^R \rightarrow [0, 1] \}_{h \in I}$ for some index set $I$ ($|I|$ is upper bounded by some constant for \globaldoublecut, some function of $a$ and $b$ for \stdoublecut, some function on $k$ on \kvc).  For at least $\beta$ fraction of $u \in U_B$,
there exist $i \in I$ and $q \in [R]$ such that $\Inf_{q}^{\leq d} [f_{u, i}] \geq \tau$. Set $l(u) = q$ for those vertices. 
Since
\begin{align*}
\Inf_{q}^{\leq d}(f_{u, i}) 
&= \sum_{\alpha_q \neq 0, |\alpha| \leq d} \widehat{f_{u, i}}(\alpha)^2 
= \sum_{\alpha_q \neq 0, |\alpha| \leq d} (\E_w [\widehat{f_{w, i}}(\pi(u, w)^{-1}(\alpha))]^2) \\
&\leq \sum_{\alpha_q \neq 0, |\alpha| \leq d} \E_w [\widehat{f_{w, i}}(\pi(u, w)^{-1}(\alpha))^2] = \E_w [\Inf_{\pi(u, w)^{-1}(q)}^{\leq d}(f_{w, i})],
\end{align*}
at least $\tau/2$ fraction of $u$'s neighbors satisfy $\Inf_{\pi(u, w)^{-1}(q)}^{\leq d} (f_{w, i}) \geq \tau / 2$. There are at most $2d / \tau$ coordinates with degree-$d$ influence at least $\tau / 2$ for a fixed $h$, so their union over $i \in I$ yields at most $2d \cdot |I| / \tau$ coordinates. Choose $l(w)$ uniformly at random among those coordinates (if there is none, set it arbitrarily). The above probabilistic strategy satisfies at least $\beta(\tau / 2)(\tau / (2d \cdot |I|) )$ fraction of all edges. Taking $\eta$ smaller than this quantity proves the soundness of the reductions. 

\paragraph*{Final Results.}
Combining our completeness and soundness analyses and taking $\epsilon$ and $\eta$ small enough, we prove our main results.

\paragraph*{\globaldoublecut.} It is hard to distinguish the following cases. 
\begin{enumerate}
\item Completeness: There is a $\{ s, t \}$-double cut of weight at most $4(1 + 2 \epsilon) / 3 + 4\eta$. 
\item Soundness: There is no global double cut of weight less than $2(1 - \epsilon)(1 - \beta)$. 
\end{enumerate}
The gap is 
\[
\frac{2(1 - \epsilon)(1 - \beta)}{\frac{4(1 + 2 \epsilon)}{3} + 4\eta}, 
\]
which approaches to $1.5$ by taking $\epsilon, \eta, \beta$ small. 
This proves Theorem~\ref{thm:NodeDoubleCut-hardness}. 

\paragraph*{\stdoublecut.} It is hard to distinguish the following cases. 
\begin{enumerate}
\item Completeness: There is a $\{ s, t \}$-double cut of weight at most $ab / (b - 2a) + ab \epsilon + ab \eta$. 
\item Soundness: There is no $\{ s, t \}$-double cut of weight less than $(2a - 1)(1 - \epsilon)(1 - \beta)$. 
\end{enumerate}
The gap is 
\[
\frac{(2a - 1)(1 - \epsilon)(1 - \beta)}{\frac{ab}{b - 2a} + ab \epsilon + ab \eta}, 
\]
which approaches to $2$ by taking $a$ large, $b$ larger, and $\epsilon, \eta, \beta$ small. 
This proves Theorem~\ref{thm:st-node-double-cut-hardness}.

\paragraph*{\kvc.} It is hard to distinguish the following cases. 
\begin{enumerate}
\item Completeness: There is a vertex cover of weight at most $k(1 + \epsilon) / 2 + k\eta$. 
\item Soundness: There is no vertex cover of weight less than $(k - 1)(1 - \epsilon)(1 - \beta)$. 
\end{enumerate}
The gap is 
\[
\frac{(k - 1)(1 - \epsilon)(1 - \beta)}{\frac{k(1 + \epsilon)}{2} + k\eta}, 
\]
which approaches to $2(k - 1) / k$ by taking $\epsilon, \eta, \beta$ small. In particular, it approaches to $4/ 3$ for $k = 3$ and $3 /2$ for $k = 4$.
Take large $r$ and small $\epsilon, \beta, \eta$. 
With Lemma~\ref{lem:node-3-cut-hardness}, this implies Theorem~\ref{thm:node-3-cut-hardness}. 
With Lemma~\ref{lem:single-terminal-bicut-hardness}, this implies Theorem~\ref{thm:single-terminal-bicut-hardness}. 
With Lemma~\ref{lem:node-bicut-hardness}, this implies Theorem~\ref{thm:node-bicut-hardness}.

\section{\textsc{EdgeLin3Cut} problems}
\label{sec:lin3cut-approx}
Given a directed graph $D=(V,E)$, a feasible solution to \srtlinthreecut\ in $D$ is a subset $F$ of arcs whose deletion from the graph eliminates all directed $s\rightarrow r$, $r\rightarrow t$ and $s\rightarrow t$ paths. One of our main tools used in the approximation algorithm for \edgebicut\ is a $3/2$-approximation algorithm for 
\stlinthreecut. We present this algorithm now.
For two sets $A,B\subseteq V$, let $\beta(A,B):=|\delta^{in}(A) \cup \delta^{in}(B)|$. 

\begin{proof}[Proof of Theorem~\ref{thm:s-star-t-lin-3-cut-algorithm}]
We first rephrase the problem in a more convenient way. 
\begin{lemma}\label{lem:st-lin-3-cut-rephrase}
\stlinthreecut\ in a directed graph $D=(V,E)$ is equivalent to 
\[\min\left\{\beta(A,B): t\in A\subset B\subseteq V-\{s\}\right\}.\]
\end{lemma}
\begin{proof}
Let $F\subseteq E$ be an optimal solution for \stlinthreecut\ in $D$ and let $(A,B):=$ argmin$\{\beta(A,B): t\in A\subset B\subseteq V-s\}$. Fix an arbitrary node $r\in B-A$. Since the deletion of $\delta^{in}(A) \cup \delta^{in}(B)$ results in a graph with no directed path from $s$ to $r$, from $r$ to $t$ and from $s$ to $t$, the edge set $\delta^{in}(A)\cup\delta^{in}(B)$ is a feasible solution to \srtlinthreecut\ in $D$, thus implying that $|F|\leq \beta(A,B)$.

On the other hand, $F$ is a feasible solution for \srtlinthreecut\ in $D$ for some $r\in V-\{s,t\}$. 
Let $A$ be the set of nodes that can reach $t$ in $D-F$, and $R$ be the set of nodes that can reach $r$ in $D-F$. Then, $F\supseteq \delta^{in}(A)$. Moreover, $F\supseteq \delta^{in}(R\cup A)$ since $R\cup A$ has in-degree 0 in $D-F$, and $s$ is not in $R \cup A$ because it cannot reach 
$r$ and $t$ in $D-F$. Therefore, taking $B=R\cup A$ we get $F\supseteq \delta^{in}(A)\cup \delta^{in}(B)$. 
 \end{proof}

Our algorithm for determining an optimal pair $(A,B):=$ argmin$\{\beta(A,B): t\in A\subset B\subseteq V-s\}$ proceeds as follows: We build a chain $\mathcal{C}$ of $\overline{s}t$-sets with the property that, for some value $k\in\mathbb{Z}_+$,
\begin{enumerate}[(i)]
\item $\mathcal{C}$ contains only cuts of value at most $k$, and \label{it:1}
\item every $\overline{s}t$-set of cut value strictly less than $k$ is in $\mathcal{C}$. \label{it:2}
\end{enumerate}

We start with $k$ being the minimum $\overline{s}t$-cut value and $\mathcal{C}$ consisting of a single minimum $\overline{s}t$-cut. In a general step, we find two $\overline{s}t$-sets: a minimum $\overline{s}t$-cut $Y$ compatible with the current chain $\mathcal{C}$, i.e. $\mathcal{C}\cup\{Y\}$ forming a chain, and a minimum $\overline{s}t$-cut $Z$ not compatible with the current chain $\mathcal{C}$, i.e. crossing at least one member of $\mathcal{C}$. These two sets can be found in polynomial time. Indeed, let $t\in C_1\subset \dots, \subset C_q\subseteq V-s$ denote the members of $\mathcal{C}$. Find a minimum cut $Y_i$ with $C_i\subseteq Y_i\subseteq V \setminus C_{i+1}$ for $i=1,\dots,q$, and choose $Y$ to be a minimum one among these cuts. Concerning $Z$, for each pair $x,y$ of nodes with $y\in C_i\subseteq V-x$ for some $i\in\{1,\dots,q\}$, find a minimum cut $Z_{xy}$ with $\{t,x\}\subseteq Z_{xy}\subseteq V-\{s,y\}$, and choose $Z$ to be a minimum one among these cuts. If $d^{in}(Y) \leq d^{in}(Z)$, then we add $Y$ to $\mathcal{C}$, and set $k$ to $d^{in}(Y)$; otherwise we set $k$ to $d^{in}(Z)$, and stop. 

Let $\mathcal{C}$ denote the chain constructed by the algorithm, and let $Y$ be an arbitrary set crossing some of its members.

\begin{claim}
$d^{in}(Y)\geq d^{in}(C)$ for all $C\in\mathcal{C}$.
\end{claim}
\begin{proof}
Suppose indirectly that $d^{in}(Y)<d^{in}(C)$ for some $C\in\mathcal{C}$. Let $\mathcal{C}'$ denote the chain consisting of those members of $\mathcal{C}$ that were added before $C$. As $C$ is a set of minimum cut value compatible with $\mathcal{C}'$, $Y$ crosses at least one member of $\mathcal{C}'$. Hence, by $d^{in}(Y)<d^{in}(C)$, the algorithm stops before adding $C$, a contradiction. 
\end{proof}

The claim implies that $\mathcal{C}$ satisfies \eqref{it:1} and \eqref{it:2} with the $k$ obtained at the end of the algorithm. Indeed, \eqref{it:1} is obvious from the construction, while \eqref{it:2} follows from the claim and the fact that $\mathcal{C}$ contains all cuts of value strictly less than $k$ that are compatible with $\mathcal{C}$.

By the above, the procedure stops with a chain $\mathcal{C}$ containing all $\overline{s}t$-sets of cut value less than $k$, and an $\overline{s}t$-set $Z$ of cut value exactly $k$ which crosses some member $X$ of $\mathcal{C}$. If the optimum value of our problem is less than $k$, then both members of the optimal pair $(A,B)$ belong to the chain $\mathcal{C}$, and we can find them by taking the minimum of $\beta(A',B')$ where $A'\subset B'$ with $A',B'\in \mathcal{C}$.

We can thus assume that the optimum is at least $k$. As $d^{in}(Z)=k$ and $d^{in}(X)\leq k$, the submodularity of the in-degree function implies $d^{in}(X\cap Z)+d^{in}(X\cup Z)\leq d^{in}(Z)+d^{in}(X) \leq 2k$. Hence at least one of $d^{in}(X\cap Z)\leq k$ and $d^{in}(X\cup Z)\leq k$ holds. As $d(X\setminus Z, X \cap Z)+d(Z\setminus X, X \cap Z)\leq d^{in}(X\cap Z)$ and $d(V\setminus (X \cup Z), X\setminus Z)+d(V\setminus (X \cup Z), Z\setminus X)\leq d^{in}(X\cup Z)$, at least one of the following four possibilities is true:

\begin{enumerate}
\item $d^{in}(X \cap Z) \leq k$ and $d(X\setminus Z, X \cap Z) \leq \frac{1}{2}k$. Choose $A=X \cap Z$, $B=X$. Then $\beta(A,B)=d(X\setminus Z,X\cap Z)+d^{in}(X)\leq \frac{1}{2}k + k = \frac{3}{2}k$.
\item $d^{in}(X \cap Z) \leq k$ and $d(Z\setminus X, X \cap Z) \leq \frac{1}{2}k$. Choose $A=X \cap Z$, $B=Z$. Then $\beta(A,B)=d(Z\setminus X,X\cap Z)+d^{in}(Z)\leq \frac{1}{2}k + k = \frac{3}{2}k$.
\item $d^{in}(X \cup Z) \leq k$ and $d(V\setminus (X \cup Z), X\setminus Z) \leq \frac{1}{2}k$. Choose $A=Z$, $B=X \cup Z$. Then $\beta(A,B)=d^{in}(Z)+d(V\setminus (X\cup Z),X\setminus Z)\leq k+ \frac{1}{2}k = \frac{3}{2}k$.
\item $d^{in}(X \cup Z) \leq k$ and $d(V\setminus (X \cup Z),Z\setminus X) \leq \frac{1}{2}k$. Choose $A=X$, $B=X \cup Z$. Then $\beta(A,B)=d^{in}(X)+d(V\setminus (X\cup Z),Z\setminus X)\leq k+ \frac{1}{2}k = \frac{3}{2}k$.
\end{enumerate}

Thus a pair $(A,B)$ can be obtained by 
taking the minimum among the four possibilities above and $\beta(A',B')$ where $A'\subset B'$ with $A',B'\in \mathcal{C}$, 
concluding the proof of the theorem.

\end{proof}

Next, we show that \stsepkcut\ is solvable in polynomial time if $k$ is a fixed constant.

Let $G=(V,E)$ be an undirected graph. Let the minimum size of an $\{s,t\}$-cut in $G$ be denoted by $\lambda_{G}(s,t)$. For two subsets of nodes $X,Y$, let $d(X,Y)$ denote the number of edges between $X$ and $Y$ and let $d(X):=d(X,V\setminus X)$. 
The cut value of a partition $\{V_1,\dots,V_q\}$ of $V$ is defined to be the total number of crossing edges, that is, $(1/2)\sum_{i=1}^q d(V_i)$, and is denoted by $\gamma(V_1,\dots,V_q)$. 
Let $\gamma^{q}(G)$ denote the value of an optimum \textsc{Edge-$q$-Cut} in $G$, i.e.,  
\[
\min\left\{\gamma(V_1,\ldots, V_q): V_i\neq \emptyset\ \forall\ i\in [q], V_i\cap V_j=\emptyset\ \forall\ i,j\in [q], \cup_{i=1}^q V_i = V\right\}.
\]

\begin{proof}[Proof of Theorem~\ref{thm:st-k-cut-algorithm}]
Let $\gamma^*$ denote the optimum value of \stsepkcut\ in $G=(V,E)$ and let $H$ denote the graph obtained from $G$ by adding an edge of infinite capacity between $s$ and $t$.
The algorithm is based on the following observation (we recommend the reader to consider $k=3$ for ease of understanding):

\begin{proposition} \label{cl:bound}
Let $\{V_1,\dots,V_k\}$ be a partition of $V$ corresponding to an optimal solution of \stsepkcut, where $s$ is in $V_{k-1}$ and $t$ is in $V_k$. Then $\gamma(V_1,\dots,V_{k-2},V_{k-1}\cup V_k)\leq 2\gamma^{k-1}(H)$.
\end{proposition}
\begin{proof}
Let $W_1,\dots,W_{k-1}$ be a minimum $(k-1)$-cut in $H$. Clearly, $s$ and $t$ are in the same part, so we may assume that they are in $W_{k-1}$. Let $U_1,U_2$ be a minimum $\{s,t\}$-cut in $G[W_{k-1}]$. Then $\{W_1,\dots,W_{k-2},\allowbreak U_1,U_2\}$ gives an $\{s,t\}$-separating $k$-cut, showing that
\begin{equation}
\gamma^*\leq\gamma(W_1,\dots,W_{k-2},U_1,U_2)=\gamma^{k-1}(H)+\lambda_{G[W_{k-1}]}(s,t). \label{eq:1}
\end{equation}



By Menger's theorem, we have $\lambda_G(s,t)$ pairwise edge-disjoint paths $P_1,\dots,\allowbreak P_{\lambda_G(s,t)}$ between $s$ and $t$ in $G$. Consider one of these paths, say $P_i$. If all nodes of $P_i$ are from $V_{k-1}\cup V_{k}$, then $P_i$ has to use at least one edge from $\delta(V_{k-1},V_k)$. Otherwise, $P_i$ uses at least two edges from $\displaystyle\delta(V_1\cup\dots\cup V_{k-2})\cup\bigcup_{\substack{i,j\leq k-2 \\ i\neq j}}\delta(V_i,V_j)$. Hence the maximum number of pairwise edge-disjoint paths between $s$ and $t$ is \[\lambda_G(s,t)\le d(V_{k-1},V_k)+\frac{1}{2}\left(d(V_1\cup\dots\cup V_{k-2})+\sum_{\substack{i,j\leq k-2 \\ i\neq j}}d(V_i,V_j)\right).\]
Thus, we have 
\begin{align*}
\gamma^*
{}&{}=
d(V_{k-1},V_k)+d(V_1\cup\dots\cup V_{k-2})+\sum_{\substack{i,j\leq k-2 \\ i\neq j}}d(V_i,V_j) \\
{}&{}\geq
\lambda_G(s,t)+\frac{1}{2}\left(d(V_1\cup\dots\cup V_{k-2})+\sum_{\substack{i,j\leq k-2 \\ i\neq j}}d(V_i,V_j)\right)\\
{}&{}=
\lambda_G(s,t)+\frac{1}{2}\gamma(V_1,\dots,V_{k-2},V_{k-1}\cup V_k)\\
{}&{}\ge 
\lambda_{G[W_{k-1}]}(s,t)+\frac{1}{2}\gamma(V_1,\dots,V_{k-2},V_{k-1}\cup V_k)
\end{align*}
that is,
\begin{equation}
\gamma^*\geq \lambda_{G[W_{k-1}]}(s,t)+\frac{1}{2}\gamma(V_1,\dots,V_{k-2},V_{k-1}\cup V_k). \label{eq:2}
\end{equation}
By combining \eqref{eq:1} and \eqref{eq:2}, we get $\gamma(V_1,\dots,V_{k-2},V_{k-1}\cup V_k)\leq 2\gamma^{k-1}(H)$, proving the proposition.

\end{proof}
Karger and Stein \cite{KS96} showed that the number of feasible solutions to \textsc{Edge-$k$-cut} in $G$ with value at most $2\gamma^{k}(G)$ is $O(n^{4k})$. All these solutions can be enumerated in  polynomial-time for fixed $k$ \cite{KS96,KM94}. This observation together with Proposition~\ref{cl:bound} gives the following algorithm for finding an optimal solution to \textsc{$\{s,t\}$-SepEdge$k$Cut}:
\begin{description}
\item[\textbf{Step 1.}] Let $H$ be the graph obtained from $G$ by adding an edge of infinite capacity between $s$ and $t$. In $H$, enumerate all feasible solutions to \textsc{Edge-$(k-1)$-Cut}---namely the vertex partitions $\{W_1,\dots,W_{k-1}\}$---whose cut value $\gamma_H(W_1,\ldots,W_{k-1})$ is at most $2\gamma^{k-1}(H)$. Without loss of generality, assume $s,t\in W_{k-1}$.
\item[\textbf{Step 2.}] For each feasible solution to \textsc{Edge-$(k-1)$-Cut} in $H$ listed in Step 1, find a minimum $\{s,t\}$-cut in $G[W_{k-1}]$, say $U_1,U_2$.
\item[\textbf{Step 3.}] Among all feasible solutions $\{W_1,\dots,W_{k-1}\}$ to \textsc{Edge-$(k-1)$-Cut} listed in Step 1 and the corresponding $U_1,U_2$ found in Step 2, return the $k$-cut $\{W_1,\dots,W_{k-2},U_1,U_2\}$ with minimum $\gamma(W_1,\ldots, W_{k-2},U_1, U_2)$. 
\end{description}

The correctness of the algorithm follows from Proposition \ref{cl:bound}: one of the choices enumerated in Step 1 will correspond to the partition $(V_1,\ldots, V_{k-2},V_{k-1}\cup V_k)$, where $(V_1,\ldots, V_k)$ is the partition corresponding to the optimal solution. 
\end{proof}


\section{Approximation for \edgebicut}
\label{sec:bicut-approx}

In this section we describe an efficient $(2-1/448)$-approximation algorithm for \edgebicut\ (Theorem \ref{thm:bicut-algorithm}).
We recall that in \edgebicut, the goal is to find the smallest number of edges in a directed graph whose deletion ensures
that there exist two distinct nodes $s$ and $t$ such that $s$ cannot reach $t$ and $t$ cannot reach $s$.
An equivalent formulation of the problem that is convenient for our purposes is as follows: Two sets $A$ and $B$ are called \emph{uncomparable} if $A\setminus B\neq \emptyset$ and $B\setminus A\neq \emptyset$.
Given a directed graph $D=(V,E)$, \edgebicut\ is equivalent to finding an uncomparable pair $A,B \subseteq V$ with minimum $|\delta^{in}(A)\cup \delta^{in}(B)|$. Indeed, if $A$ and $B$ are uncomparable and we remove $\delta^{in}(A)\cup \delta^{in}(B)$ from the directed graph, then nodes in $A \setminus B$ cannot reach nodes in $B \setminus A$ and vice versa. On the other hand, if $s$ cannot reach $t$ and $t$ cannot reach $s$, then the set of nodes that can reach $s$ and the set of nodes that can reach $t$ are uncomparable, and have in-degree $0$.

We introduce some definitions and notation in order to describe the $(2-\e)$-approximation algorithm, where the value of $\e$ is computed at the end of the proof.
\begin{definition}
For $A,B \subseteq V$, let $\bicut(A,B) := |\delta^{in}(A)\cup \delta^{in}(B)|$ and let $\cross(A,B) := d^{in}(A)+d^{in}(B)$. Furthermore, let
\begin{align*}
\OPT &:=\min\set{ \bicut(A,B)\ |\ \text{$A$ and $B$ are uncomparable}},\\
\cross &:=\min\set{ \cross(A,B)\ |\ \text{$A$ and $B$ are uncomparable}}.
\end{align*}
The pair where the latter value is attained is called the \emph{minimum uncomparable cut-pair}.
\end{definition}

\begin{definition}
If $c$ is a capacity function on a directed graph $D$, then $d^{in}_c(U) = \sum_{e\in \delta^{in}(U)} c(e)$ is the sum of the capacities of incoming edges of $U$. Similarly, $d^{out}_c(U) = \sum_{e\in \delta^{out}(U)} c(e)$. For two disjoint set of vertices $A$ and $B$, the number of edges from $A$ to $B$ is defined as $d(A,B) = |\delta^{out}(A)\cap\delta^{in}(B)|$.
\end{definition}

Clearly, $\cross(A,B) \geq \bicut(A,B)$ for any $A,B$. The following lemma shows that $\cross$ can be computed efficiently.
This means that we immediately have a $(2-\e)$-approximation if $\cross \leq (2-\e)\OPT$.

\begin{lemma}\label{lem:cross}
For a directed graph $D=(V,E)$, there exists a polynomial time algorithm to find a minimum uncomparable cut-pair. 
\end{lemma}
\begin{proof}
 For fixed vertices $a$ and $b$, there is an efficient algorithm to find $A$ and $B$ such that $a\in A\setminus B$
and $b\in B\setminus A$ and $\cross(A,B)$ is minimized. Indeed, this is precisely finding
the sink side of an $a-b$ min-cut and that of a $b-a$ min-cut. Trying all possible $a$ and $b$ and taking the minimum
gives the desired result.
\end{proof}

We also need the following lemma showing that we can minimize $\bicut(A,B)$ among pairs whose intersection is fixed.

\begin{lemma}\label{lem:fixedZ}
Given a directed graph $D=(V,E)$ and $Z\subseteq V$, there exists a polynomial time algorithm to find an uncomparable pair $A,B$ satisfying $A\cap B=Z$ that minimizes $\bicut(A,B)$ among pairs with this property.
\end{lemma}
\begin{proof}
Let $D' = D[V\setminus Z]$ be the directed graph induced on $V\setminus Z$.
The \textsc{EdgeDoubleCut} problem can be solved in polynomial time in $D'$ \cite{BP13}; let $X'$ and $Y'$ be the disjoint sets whose
incoming edges give the optimal double cut.
We claim that the pair $X'\cup Z,Y'\cup Z$ forms a minimum bicut among all bicuts with intersection $Z$.
Indeed, assume the optimal solution is $\bicut(A,B)$. Let $X=A\setminus B$, $Y=B\setminus A$ and $W = V-(A\cup B)$. Then
\begin{align*}
\bicut(X'\cup Z, Y'\cup Z) &= d^{in}_{D'}(X')+d^{in}_{D'}(Y')+d^{in}(Z)\\
                           &\leq d^{in}_{D'}(X) + d^{in}_{D'}(Y) + d^{in}(Z)\\
                           &= d^{in}(Z)+d(W,X)+d(W,Y)+d(X,Y)+d(Y,X)\\
                           &=\bicut(A,B). \tag*{}
\end{align*}
\end{proof}

A similar argument but using the complements of $A$ and $B$ proves the following.
\begin{lemma}\label{lem:fixedW}
Given a directed graph $D=(V,E)$ and $W\subseteq V$, there exists a polynomial time algorithm to find an uncomparable pair $A,B$ satisfying $V\setminus (A\cup B)=W$ that minimizes $\bicut(A,B)$ among pairs with this property. 
\end{lemma}

\begin{proof}[Proof of Theorem \ref{thm:bicut-algorithm}]
If the minimum bicut $(A,B)$ has the property that either $|A\cap B|\leq 2$ or $|V\setminus(A\cup B)|\leq 2$, then the optimal
bicut can be found by applying the above algorithms by setting $Z$ or $W$ to every possible choice of subsets of nodes of size at most $2$ and considering the minimum. The algorithm that we present
below makes use of this observation.

\medskip

\hrule
\begin{tabbing}
        \quad\=\quad\=\quad\=\quad\=\quad\=\quad\=\quad\=\quad\=\quad\=\quad\=\quad\=\quad\=\quad\=\kill
\textsc{ApproximateGlobalBicut} for directed graph $D=(V,E)$\+
\\   Find minimum bicut if $|Z|\leq 2$ or $|W|\leq 2$ using Lemmas \ref{lem:fixedZ} and \ref{lem:fixedW}
\\   Compute the minimum uncomparable cut-pair
\\   For each tuple of nodes $(x,y,w_1,w_2,z_1,z_2)$\+
\\   $X' \gets$ sink-side of the minimum $\{w_1,w_2,y\} \to \{x,z_1,z_2\}$-cut
\\   $Y' \gets$ sink-side of the minimum $\{w_1,w_2,x\} \to \{y,z_1,z_2\}$-cut
\\   $E_1 \gets E[X'] \cup E[Y']$
\\   $E_2 \gets E[V\setminus X'] \cup E[V \setminus Y']$
\\   $D_1\gets $ $D$ with the arcs in $E_1$ duplicated
\\   $D_2\gets $ $D$ with the arcs in $E_2$ duplicated
\\   $Z' \gets$ sink-side of minimum $\{w_1,w_2,x,y\} \to \{z_1,z_2\}$-cut in $D_1$
\\   $W' \gets$ source-side of minimum $\{w_1,w_2\} \to \{x,y,z_1,z_2\}$-cut in $D_2$
\\   $D' \gets$ contract $X' \cap Y'$ to $z'$, contract $V \setminus X'$ to $w'$, remove all $w'z'$ arcs
\\   In $D'$, find $\overline{w'}z'$-sets $A' \subsetneq B'$ such that $\bicut(A',B')$ is \+
\\    at most 3/2 times minimum, using Theorem \ref{thm:s-star-t-lin-3-cut-algorithm} and Lemma \ref{lem:st-lin-3-cut-rephrase} \-
\\   Find all bicuts that can be generated using set operations on \+
\\    $X',Y',Z',W',A',B'$.\-\-
\\   Output the minimum bicut among all the bicuts found
\-
\end{tabbing}
\hrule
\medskip

To show that the algorithm is correct, first we fix a global min-bicut $(A,B)$, i.e. the optimum is $\OPT=\bicut(A,B)$. Let $X=A\setminus B$, $Y=B\setminus A$, $Z=A\cap B$,
and $W = V\setminus(A\cup B)$. We assume that both $Z$ and $W$ are of size at least $3$, otherwise it is clear that the algorithm finds the optimum.
Let $\e>0$ be a constant whose value will be determined later. 

\begin{lemma}
\label{lem:crossworks}
If one of the following is true, then the minimum uncomparable cut-pair is a $(2-\e)$-approximation:
\begin{enumerate}[(i)]
\item $d(Z,W) \leq (1-\e)\OPT$, \label{item:dZW}
\item For every $z_1, z_2 \in Z$, there exists a subset $U$ of nodes containing $z_1, z_2$ but not $Z$ with $d^{in}(U) < (1-\e)\OPT$.
\item For every $w_1, w_2 \in W$, there exists a subset $U$ of nodes not containing $w_1, w_2$ but intersecting $W$ with $d^{in}(U) < (1-\e)\OPT$.
\end{enumerate}
\end{lemma}
\begin{proof} \mbox{}
\begin{enumerate}[(i)]
\item The pair $(A,B)$ is uncomparable, and $\cross(A,B)\leq(2-\e)\OPT$ if $d(Z,W)\leq (1-\e)\OPT$. Therefore the
minimum uncomparable cut-pair is a $(2-\e)$-approximation if \eqref{item:dZW} holds.
\item Among the sets with in-degree less than $(1-\e)\OPT$ which do not contain every node of $Z$, let $T$ be the one with
inclusionwise maximal intersection with $Z$. Let $z_1 \in Z \setminus T$ and $z_2 \in Z \cap T$.
There exists a set $U$ such that $d^{in}(U) < (1-\e)\OPT$ and $z_1,z_2\in U$ that contains $z_1$ and $z_2$ but not the whole $Z$. Because of the maximal intersection of $T$ with $Z$, we have that $T\not\subset U$. Hence $T$ and $U$ are uncomparable and $\cross(T,U) \leq (2-2\e)\OPT$.
\item Argument similar to the proof of (ii) shows that the minimum uncomparable cut-pair is a $(2-2\e)$-approximation if \emph{(iii)} holds. 
\end{enumerate}
\end{proof}

By Lemma \ref{lem:crossworks}, we only have to consider the case where $d(W,Z)\geq (1-\e)\OPT$. We will also assume that $z_1,z_2,w_1,w_2$ are chosen such that $d^{in}(U)\geq (1-\e)\OPT$ for all subsets $U$ of nodes containing $z_1,z_2$ but not $Z$, 
and $d^{in}(U)\geq (1-\e)\OPT$ for all subsets $U$ of nodes not containing $w_1,w_2$ but intersecting $W$. 

We may assume that $\bicut(X',Y')\geq (2-\e)\OPT$, otherwise $X'$ and $Y'$ forms a $(2-\e)$-approximation. Also, $d^{in}(X')\leq d^{in}(X\cup Z)\leq \OPT$ because $X'$ is the sink-side of a min $\set{w_1,w_2,y}\to \set{x,z_1,z_2}$ cut. Similarly, $d^{in}(Y')\leq d^{in}(Y\cup Z) \leq \OPT$.
Let $c$ be the capacity function obtained by increasing the capacity of each edge in $E_1$ to $2$, and
let $\bar{c}$ be the capacity function obtained by increasing the capacity of each edge in $E_2$ to $2$.
We consider four cases depending on the relations between $W$ and $X'\cup Y'$, and between $Z$ and $X'\cap Y'$.

\paragraph*{Case 0.} Suppose $W\cap (X'\cup Y') = \emptyset$, $Z \subseteq X'\cap Y'$.
In this case $\delta^{in}(X')$ and $\delta^{in}(Y')$ both contain all edges counted in $d(W,Z)$. Hence $\bicut(X',Y') \leq \cross(X',Y')-d(W,Z) \leq (1+\e)\OPT$. This shows that $(X',Y')$ is a $(1+\e)$-approximation. \\

For the remaining three cases, we will use the following proposition. 
\begin{proposition}
\label{cl:quick2}
If $d^{in}(X'\cap Z') \geq (1-\e)\OPT$ and  $d^{in}(Y'\cap Z') \geq (1-\e)\OPT$, then
$\bicut(X' \cup Z',Y' \cup Z') \leq 2\e\OPT + d_c^{in}(Z)$.
\end{proposition}
\begin{proof}
If $d^{in}(X'\cap Z')\geq (1-\e)\OPT$, then $d^{in}(X')-d^{in}(X'\cap Z') \leq \e\OPT$, so
    \begin{align}
        d^{in}(X'\cup Z') &\leq d^{in}(Z')+\e \OPT - d(X'\setminus Z',Z'\setminus X') - d(Z'\setminus X',X'\setminus Z')\notag\\
        &\leq d^{in}(Z')+\e \OPT - d(X'\setminus Z',Z'\setminus X').\label{eq:tmp1}
    \end{align}
Similarly,
\begin{equation}
d^{in}(Y'\cup Z') \leq d^{in}(Z')+\e \OPT - d(Y'\setminus Z',Z'\setminus Y'). \label{eq:tmp2}
\end{equation}
We need the following claim.
\begin{claim}
\begin{multline}
 \bicut(X'\cup Z',Y'\cup Z')
 \leq \cross(X' \cup Z',Y' \cup Z')+d^{in}_c(Z')-2d^{in}(Z')\\+d(X'\setminus Z',Z'\setminus X')+d(Y'\setminus Z',Z'\setminus Y')\label{eq:tmp3}.
\end{multline}
\end{claim}
\begin{proof}
By counting the edges entering $Z'$, we have
\begin{enumerate}
\item $d^{in}_c(Z') = d^{in}(Z')+|\delta^{in}(Z')\cap E_1|$.
\item $d^{in}(Z') = d(V\setminus (X'\cup Y'\cup Z'),Z') + |\delta^{in}(Z')\cap E_1| + d(X'\setminus Z', Z'\setminus X') + d(Y'\setminus Z', Z'\setminus Y') - d((X'\cap Y')\setminus Z', Z'\setminus(X'\cup Y'))$.
\end{enumerate}
The first equation can be rewritten as $d^{in}_c(Z') - 2d^{in}(Z') = -d^{in}(Z') + |\delta^{in}(Z')\cap E_1|$. Using this and the second equation, we get
 $d^{in}_c(Z')-2d^{in}(Z')+d(X'\setminus Z', Z'\setminus X')+d(Y'\setminus Z', Z'\setminus Y') = -d(V\setminus (X'\cup Y'\cup Z'),Z') + d((X'\cap Y')\setminus Z', Z'\setminus(X'\cup Y'))$.
Thus the desired inequality (\ref{eq:tmp3}) simplifies to
\begin{multline*}
\bicut(X'\cup Z',Y'\cup Z') \leq \cross(X' \cup Z',Y' \cup Z')\\-d(V\setminus (X'\cup Y'\cup Z'),Z')+d((X'\cap Y')\setminus Z', Z'\setminus(X'\cup Y')).
\end{multline*}
The elements counted by $d(V\setminus (X'\cup Y'\cup Z'),Z')$ are counted twice in $\cross(X' \cup Z',Y' \cup Z')$. Hence we have the desired relation (\ref{eq:tmp3}).
\end{proof}
    Using (\ref{eq:tmp1}),(\ref{eq:tmp2}) and (\ref{eq:tmp3}) we get
    \begin{align*}
    \bicut&(X' \cup Z',Y' \cup Z') \\
    &\leq d^{in}(X' \cup Z')+d^{in}(Y' \cup Z')\\ &\qquad +d^{in}_c(Z')-2d^{in}(Z')+d(X'\setminus Z',Z'\setminus X')+d(Y'\setminus Z',Z'\setminus Y')\\
    &\leq d^{in}(Z')+\e \OPT + d^{in}(Z') + \e\OPT + d^{in}_c(Z')-2d^{in}(Z')\\
    &\leq 2\e\OPT + d_c^{in}(Z')\\
    &\leq 2\e\OPT + d_c^{in}(Z). \tag*{}
    \end{align*}
\end{proof}

\paragraph*{Case 1.} Suppose $W\cap (X' \cup Y') = \emptyset$ and $Z\not\subseteq X' \cap Y'$. Without loss of generality, let $Z\not \subseteq X'$. The set $X'\cap Z'$ contains $z_1,z_2$ but not the whole $Z$, hence $d^{in}(X'\cap Z')\geq (1-\e)\OPT$.

We first consider the subcase when $d^{in}(Y' \cap Z') \geq (1-\e)\OPT$ also holds. Then, by Proposition \ref{cl:quick2}, we get
\[
\bicut(X' \cup Z',Y' \cup Z') \leq 2\e\OPT + d_c^{in}(Z).
\]
We are in the case where $(X'\cup Y')\cap W=\emptyset$, so $d^{in}_c(Z) \leq d^{in}(Z) + d(X,Z) + d(Y,Z)\leq (1+\e)\OPT$ since $d(W,Z)\geq (1-\e)\OPT$. Hence we have $\bicut(X' \cup Z',Y' \cup Z')\leq (1+3\e)\OPT$.

Next we consider the other subcase where $d^{in}(Y' \cap Z') < (1-\e)\OPT$. By the choice of $z_1,z_2$, this means that $Z\subseteq Y' \cap Z'$. In this case $Y' \cap Z'$ crosses $X'$, because $X'$ does not contain the whole $Z$, and $Y'\cap Z'$ does not contain $x$ and thus $X'$ and $Y'\cap Z'$ are uncomparable. Since $\cross(X',Y'\cap Z') \leq (2-\e)\OPT$, the minimum uncomparable cut-pair is a $(2-\e)$-approximation.

\paragraph*{Case 2.} Suppose $W\cap (X' \cup Y') \neq \emptyset$ and $Z\subseteq X' \cap Y'$. This is similar to Case 1 by symmetry.

\paragraph*{Case 3.} Suppose $W\cap (X' \cup Y') \neq \emptyset$ and $Z\not \subseteq X' \cap Y'$.

We may assume that $Z\not\subseteq X'$ without loss of generality. Because of the choice of $z_1,z_2$, we have $d^{in}(X'\cap Z')\geq (1-\e)\OPT$. By the same argument as in Case 1 (last paragraph), we may assume that $d^{in}(Y'\cap Z')\geq (1-\e)\OPT$ as well. The inequality $\bicut(X'\cup Z',Y'\cup Z') \leq 2\e\OPT + d_c^{in}(Z)$ holds using Proposition \ref{cl:quick2}. If $d^{in}_c(Z) \leq (2-3\e)\OPT$, then these imply $\bicut(X'\cup Z',Y'\cup Z') \leq (2-\e)\OPT$. Similarly, if $d^{out}_{\overline{c}}(W) \leq (2-3\e)\OPT$, then we obtain
$\bicut(X'\setminus  W',Y'\setminus W') \leq (2-\e)\OPT$. Thus, we may assume that both $d^{in}_c(Z)$ and $d^{out}_{\overline{c}}(W)$ are at least $(2-3\e)\OPT$.

\begin{figure}
\begin{center}
\includegraphics[scale=0.5]{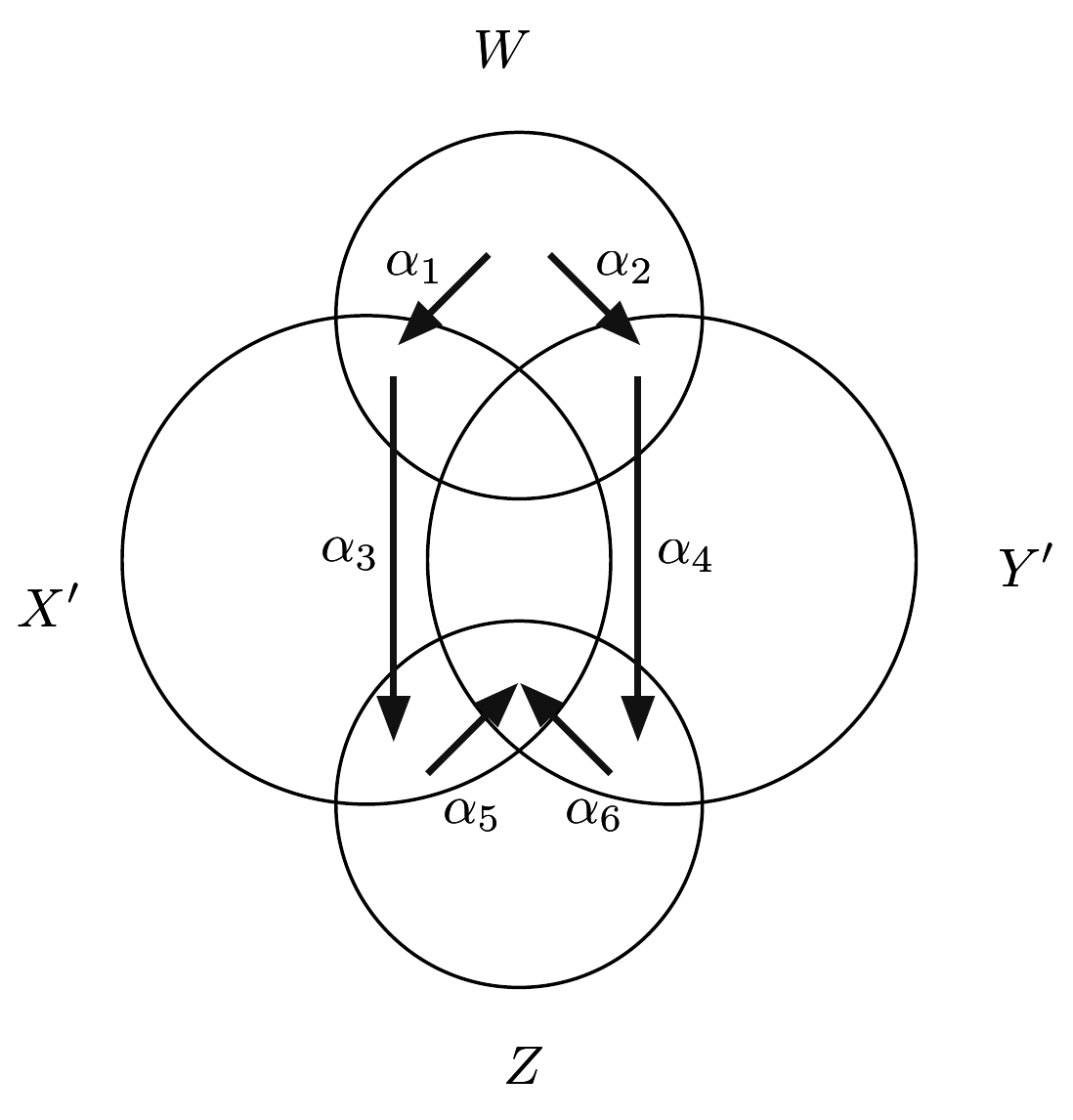}
\end{center}
\caption{The quantities $\alpha_1,\ldots,\alpha_6$.}
\label{xyzw_diagram}
\end{figure}

Let us define the following quantities (see Figure \ref{xyzw_diagram}).
\begin{enumerate}
\item $\alpha_1=d(W\setminus (X'\cup Y'),W\cap (X'\setminus Y'))$,
\item $\alpha_2=d(W\setminus (X'\cup Y'),W\cap (Y'\setminus X'))$,
\item $\alpha_3=d(W\cap (X'\setminus Y'),Z\cap (X'\setminus Y'))$,
\item $\alpha_4=d(W\cap (Y'\setminus X'),Z\cap (Y'\setminus X'))$,
\item $\alpha_5=d(Z\cap (X'\setminus Y'),X'\cap Y'\cap Z')$, and
\item $\alpha_6=d(Z\cap (Y'\setminus X'),X'\cap Y'\cap Z')$.
\end{enumerate}
In propositions \ref{cl:union}, \ref{cl:34sum}, \ref{cl:1256sum}, \ref{cl:1625sum}, \ref{cl:36sum} and \ref{cl:flows2}, 
we show a sequence of inequalities involving these quantities.


\begin{proposition}
\label{cl:union}
$(1-\e)\OPT \leq d^{in}(X'\cap Y'),d^{in}(X'\cup Y'),d^{in}(X'\cap Z),d^{in}(X'\cup Z)\leq (1+\e)\OPT$.
\end{proposition}
\begin{proof}
By submodularity, $d^{in}(X'\cap Y')+d^{in}(X'\cup Y')\leq d^{in}(X')+d^{in}(Y')\leq 2\OPT$.
We note that $d^{in}(X'\cap Y')\geq (1-\e)\OPT$ by the choice of $z_1,z_2$. This shows $d^{in}(X'\cup Y')\leq (1+\e)\OPT$.
Similarly, $d^{in}(X'\cup Y')\geq (1-\e)\OPT$ by the choice of $w_1,w_2$, and hence $d^{in}(X'\cap Y')\leq (1+\e)\OPT$.

By our assumption, $X'$ and $Z$ are uncomparable, hence $X'\cap Z$ contains both $z_1,z_2$ but not all of $Z$. By the choice of $z_1,z_2$, $d^{in}(X'\cap Z) \geq (1-\e)\OPT$. By submodularity,
$d^{in}(X'\cup Z) \leq d^{in}(X')+d^{in}(Z)-d^{in}(X'\cap Z) \leq (1+\e)\OPT$. For the remaining inequalities,
we notice that $X' \cup Z$ and $Y'$ are uncomparable, so $\cross(X'\cup Z, Y')\geq (2-\e)\OPT$, and therefore
$d^{in}(X'\cup Z)  \geq (1-\e)\OPT$. Submodularity gives $d^{in}(X'\cap Z) \leq (1+\e)\OPT$.
\end{proof}

\begin{proposition}
\label{cl:34sum}
$(1-6\e)\OPT \leq \alpha_3+\alpha_4\leq \OPT$.
\end{proposition}
\begin{proof}
We recall that $d_c^{in}(Z) = d^{in}(Z)+|\delta^{in}(Z)\cap E_1| \geq (2-3\e)\OPT$ and
$d_{\bar{c}}^{out}(W) = d^{out}(W)+|\delta^{out}(W)\cap E_2| \geq (2-3\e)\OPT$.
Let $C$ be the set of edges from $W$ to $Z$, i.e.\ those counted by $d(W,Z)$. Let $a=|\delta^{in}(Z)\setminus C|$ and $b=|\delta^{out}(W)\setminus C|$. 
We note that $\alpha_3+\alpha_4=|C\cap E_1\cap E_2|$ and $|C|+a+b\leq \OPT$.

We have $|C\cap E_1|\geq |\delta^{in}(Z)\cap E_1| - a$ and $|C\cap E_2|\geq |\delta^{out}(W)\cap E_2| - b$.
From all the above, we get the following sequence of inequalities.
\begin{align*}
&|C\cap E_1\cap E_2| \\
\geq& |C| - |C\setminus E_1| - |C\setminus E_2| \\
\geq& |C| - (|C| - (|\delta^{in}(Z)\cap E_1| - a)) - (|C| - (|\delta^{out}(W)\cap E_2| - b))\\
=&|\delta^{in}(Z)\cap E_1| + |\delta^{out}(W)\cap E_2| -|C| - a - b\\
\geq &(2-3\e)\OPT - d^{in}(Z) + (2-3\e)\OPT - d^{out}(W) -|C| - a - b\\
\geq &(4-6\e)\OPT - 3\OPT\\
= &(1-6\e)\OPT . \tag*{}
\end{align*}
\end{proof}

\begin{proposition}
\label{cl:1256sum}
$(1-8\e)\OPT \leq \alpha_1+\alpha_2\leq (1+\e)\OPT$ and $(1-8\e)\OPT \leq \alpha_5+\alpha_6\leq (1+\e)\OPT$.
\end{proposition}
\begin{proof}
The upper bounds  follow from $\alpha_1+\alpha_2\leq d^{in}(X'\cup Y')\leq (1+\e)\OPT$ and $\alpha_5+\alpha_6\leq d^{in}(X'\cap Y')\leq (1+\e)\OPT$.
For the lower bound, we observe that $d^{in}(X'\cap Y'\cap Z) \geq (1-\e)\OPT$ but $|\delta^{in}(X')\cap \delta^{in}(Y')| \leq \e\OPT$.
This shows that at most $\e\OPT$ edges enter $X'\cap Y'\cap Z$ from outside of $X'\cup Y'$. By Proposition \ref{cl:34sum}, 
$|\delta^{in}(Z)\cap \delta^{in}(X'\cap Y'\cap Z)| \leq \delta^{in}(Z) - \alpha_3-\alpha_4 \leq 6\e\OPT$.
These imply $\alpha_5+\alpha_6 \geq (1-\e)\OPT-6\e\OPT-\e\OPT = (1-8\e)\OPT$.

A similar argument shows the lower bound for $\alpha_1+\alpha_2$. The choice of $w_1,w_2$ implies $d^{out}(W\setminus(X'\cup Y')) \geq (1-\e)\OPT$,
and at most $\e\OPT$ edges enter $X'\cap Y'$ from $W \setminus(X'\cup Y')$. By Proposition \ref{cl:34sum},  
$|\delta^{out}(W)\cap \delta^{out}(W\setminus (X'\cup Y'))| \leq \delta^{out}(W) - \alpha_3-\alpha_4 \leq 6\e\OPT$.
These imply $\alpha_1+\alpha_2 \geq (1-\e)\OPT-6\e\OPT-\e\OPT = (1-8\e)\OPT$.
\end{proof}

\begin{proposition}
\label{cl:1625sum}
$(1-16\e)\OPT\leq \alpha_1+\alpha_6\leq \OPT$ and $(1-16\e)\OPT\leq \alpha_2+\alpha_5\leq \OPT$.
\end{proposition}
\begin{proof}
The upper bounds follow by $\alpha_1+\alpha_6\leq d^{in}(X') \leq \OPT$ and $\alpha_2+\alpha_5\leq d^{in}(Y') \leq \OPT$.
On the other hand, combining the two inequalities in Proposition \ref{cl:1256sum} gives $(2-16\e)\OPT \leq \alpha_1+\alpha_2+\alpha_5+\alpha_6$,
so we also get the lower bounds.
\end{proof}

\begin{proposition}
\label{cl:36sum}
$(1-23\e)\OPT \leq \alpha_3+\alpha_6 \leq (1+\e)\OPT$.
\end{proposition}
\begin{proof}
Consider the set $T=X'\cap Z$. By Proposition \ref{cl:union}, we have $\alpha_3+\alpha_6 \leq d^{in}(T)\leq (1+\e)\OPT$, which gives the upper bound.
Also, $\alpha_1+\alpha_6+d(Z\setminus (X'\cup Y'),T)\leq d^{in}(X')\leq \OPT$, hence $d(Z\setminus (X'\cup Y'),T)\leq 16\e\OPT$.
Since $\alpha_3+\alpha_4\geq (1-6\e)\OPT$, the remaining contribution to $d^{in}(Z)$ from elsewhere is at most $6\e\OPT$. We obtain
\begin{align*}
(1-\e)\OPT &\leq d^{in}(T)\\
           &\leq \alpha_6+d(Z\setminus (X'\cup Y'),T)+d(V\setminus Z,T')\\
           &\leq \alpha_6+16\e\OPT +\alpha_3+ 6\e\OPT\\
           &\leq \alpha_3+\alpha_6 +22\e\OPT.
\end{align*}\\
Hence, $(1-23\e)\OPT \leq \alpha_3+\alpha_6$.
 \end{proof}

\begin{proposition}
\label{cl:flows2}
$\alpha_1+\alpha_5\geq 2\alpha_3-51\e\OPT$.
\end{proposition}
\begin{proof}
The above claims give us a chain of relations:
\begin{align*}
(1-16\e)\OPT - \alpha_6  &\leq \alpha_1   \leq \OPT-\alpha_6,\\
(1-8\e) \OPT - \alpha_1  &\leq \alpha_2   \leq (1+\e)\OPT-\alpha_1,\\
(1-16\e)\OPT - \alpha_2  &\leq \alpha_5   \leq \OPT-\alpha_2,\\
(1-23\e)\OPT - \alpha_3  &\leq \alpha_6   \leq (1+\e)\OPT-\alpha_3.\\
\end{align*}
By substitution, we get
\begin{align*}
\alpha_3-17\e\OPT &\leq \alpha_1\leq \alpha_3 + 23\e\OPT,\\
\alpha_1-17\e\OPT &\leq \alpha_5\leq \alpha_1+8\e\OPT,\\
\alpha_3-34\e\OPT &\leq \alpha_5\leq \alpha_1+31\e\OPT.
\end{align*}
Therefore $\alpha_1+\alpha_5\geq 2\alpha_3-51\e\OPT$.
\end{proof}

Without loss of generality, let $\alpha_3\geq (\alpha_3+\alpha_4)/2$, since if not, there is another iteration of the algorithm where $x$ and $y$ are switched. Therefore, $\alpha_3\geq (1/2-3\e)\OPT$.

Let $A_0=(X'\cap Z) \cup (X'\cap Y')$ and $B_0=(X'\setminus W) \cup (X'\cap Y')$. Let $H$ be the directed graph obtained by
contracting $X' \cap Y'$ to a node $z'$,
contracting $V\setminus X'$ to a node $w'$, and removing all $w'z'$ arcs.

\begin{claim}
$|\delta^{in}_H(A_0) \cup \delta^{in}_H(B_0)|\leq \alpha_3 + 39\e\OPT$.
\end{claim}
\begin{proof}
The left hand side is precisely
\begin{multline*}
d(V\setminus A_0,(X'\cap Z)\setminus Y')+d(X'\setminus(Y'\cup Z),X'\cap Y')\\+d((V\setminus X') \cup (X'\cap W\setminus Y')),X'\setminus(Y'\cup W\cup Z)).
\end{multline*}
We consider each term individually.

\begin{enumerate}
\item The term $d(V\setminus A_0,(X'\cap Z)\setminus Y')$ counts a subset of the edges going into $Z$. All but $\e\OPT$ edges are from $W$. Hence $d((V\setminus A_0)\cap W,(X'\cap Z)\setminus Y') \geq d(V\setminus A_0,(X'\cap Z)\setminus Y')-\e\OPT$. All but at most $6\e\OPT$ edges counted by $d((V\setminus A_0)\cap W,(X'\cap Z)\setminus Y')$ are counted by $\alpha_3$. This means
$d((V\setminus A_0)\cap W,(X'\cap Z)\setminus Y') \leq \alpha_3 + 6\e\OPT$. In total, it shows $d(V\setminus A_0,(X'\cap Z)\setminus Y')\leq \alpha_3+7\e\OPT$.
\item The term $d(X'\setminus(Y'\cup Z),X'\cap Y')$ counts a subset of the edges going into $Y'$. We have $d^{in}(Y')\geq d(X'\setminus(Y'\cup Z),X'\cap Y')+\alpha_2+\alpha_5$. Using Proposition \ref{cl:1625sum}, we obtain $d(X'\setminus(Y'\cup Z),X'\cap Y')\leq 16\e\OPT$.
\item Let $t=d((V\setminus X') \cup (X'\cap W\setminus Y')),X'\setminus(Y'\cup W\cup Z))$. So $t$ counts a subset of the edges going into $X'$. Using Proposition \ref{cl:1625sum} and $d^{in}(X')\geq t+\alpha_1+\alpha_6$, we obtain $t\leq 16\e\OPT$.
\end{enumerate}
Thus, the total contribution is at most $\alpha_3+39\e\OPT$.
\end{proof}

Using the 3/2-approximation algorithm of Theorem \ref{thm:s-star-t-lin-3-cut-algorithm} and Lemma \ref{lem:st-lin-3-cut-rephrase}, we can find $\overline{w'}z'$-sets $A' \subsetneq B'$ such that
$|\delta^{in}_H(A') \cup \delta^{in}_H(B')|\leq 3(\alpha_3+39\e\OPT)/2$.
Now consider $\bicut(X' \cap B',Y' \cup A')$ in the original directed graph. We have the following inequality by counting the edges on the left hand side (see Figure \ref{final_ineqremove} for a proof).
\begin{align}
\bicut(X' \cap B',Y' \cup A')
+\alpha_5 + \alpha_1 \leq \cross(X',Y') + |\delta^{in}_H(A')\cup\delta^{in}_H(B')|.
\label{eq:final_ineq}
\end{align}
Using Proposition \ref{cl:flows2} and the assumption that $\alpha_3\geq (1/2-3\e)\OPT$, we get  
\begin{align*}
\bicut(X' \cap B',Y' \cup A')  &\leq \cross(X',Y') + |\delta^{in}_H(A')\cup\delta^{in}_H(B')| -\alpha_5 - \alpha_1\\
                              &\leq 2\OPT + \frac{3}{2}\alpha_3 + 117\e\OPT/2 - \alpha_5 - \alpha_1\\
                              &\leq 2\OPT + \frac{3}{2}\alpha_3 + 117\e\OPT/2 - (2\alpha_3-51\e\OPT)\\
                              &\leq 2\OPT - \frac{1}{2}\alpha_3 + 117\e\OPT/2 + 51\e\OPT\\
                              &\leq (2+(219/2)\e)\OPT - (\frac{1}{2}-3\e)\OPT)/2\\
                              &= (7/4+111\e)\OPT .\\
\end{align*}

\begin{figure}
\begin{center}
\includegraphics[scale=0.6]{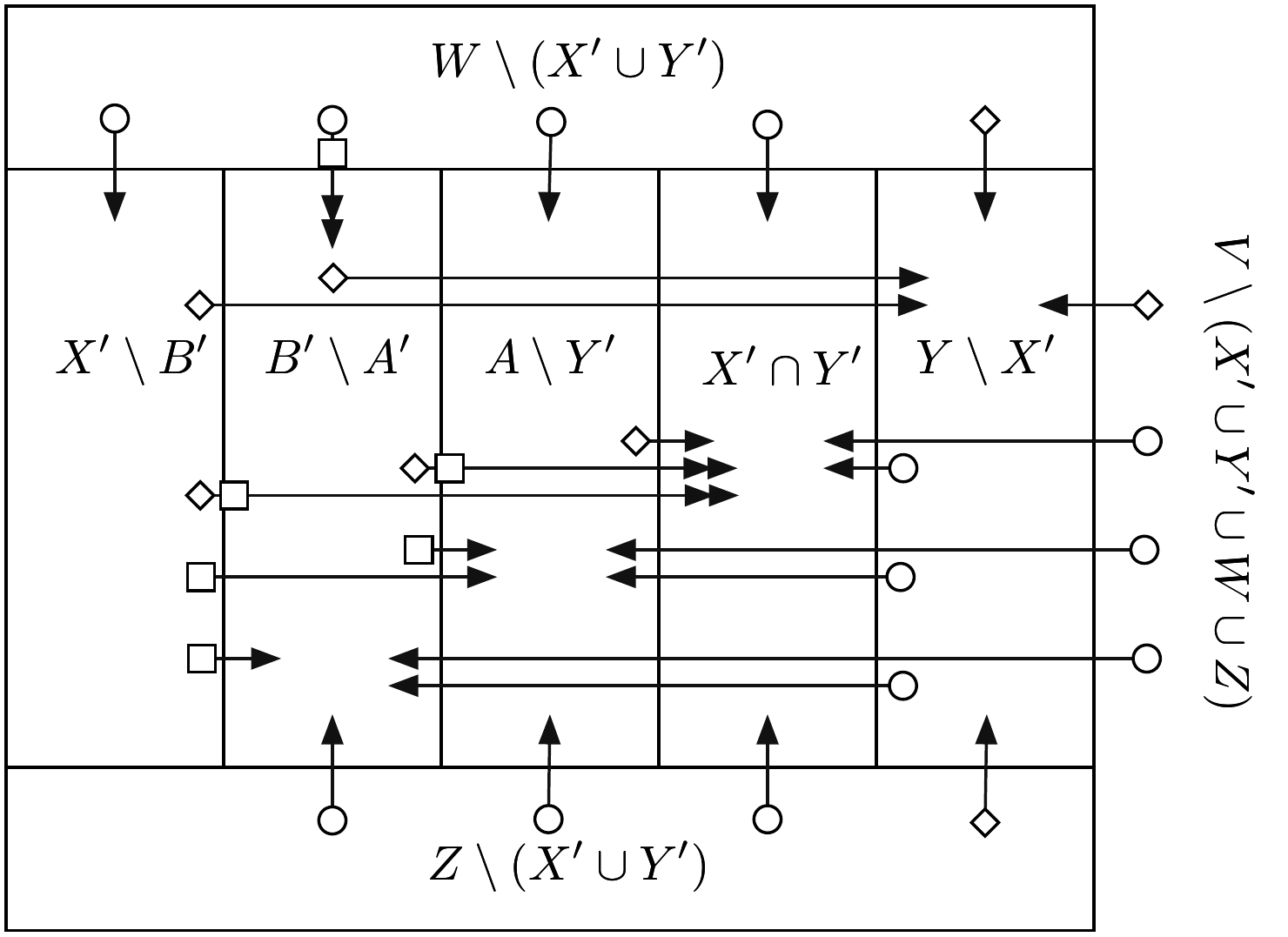}
\end{center}
\caption{Every arrow indicates the edges counted in the left hand side of \eqref{eq:final_ineq}. $\circ$ indicates that the edge is in $\delta^{in}(X')$, $\diamond$ indicates that the edge is in $\delta^{in}(Y')$ and $\square$ indicates that the edge is in $\delta^{in}_H(A')\cup\delta^{in}_H(B')$. Irrelevant edges are not shown. }
\label{final_ineqremove}
\end{figure}
Based on all the cases analyzed above, the approximation factor is $\max\set{1+\e,2-\e,7/4+111\e}$. In order to minimize the factor, we set $\e=1/448$ to get the desired approximation factor. 
\end{proof}

\vspace{2mm}
\noindent \textbf{Acknowledgements.} Karthik would like to thank Chandra Chekuri, Neil Olver and Chaitanya Swamy for helpful discussions at various stages of this work. 
\bibliographystyle{amsplain}
\bibliography{references}

\end{document}